\documentclass[letterpaper]{article} 
\usepackage{arxiv}
\usepackage{times}  
\usepackage{helvet}  
\usepackage{courier}  
\usepackage[hyphens]{url}  
\usepackage{graphicx} 
\usepackage{caption} 
\usepackage{algorithm}
\usepackage{algorithmic}

\usepackage{newfloat}
\usepackage{listings}
\DeclareCaptionStyle{ruled}{labelfont=normalfont,labelsep=colon,strut=off} 
\floatstyle{ruled}
\newfloat{listing}{tb}{lst}{}
\floatname{listing}{Listing}
\usepackage{amsmath, amsthm, amssymb}
\newtheorem{example}{Example}
\newtheorem{thm}{Theorem}
\newtheorem{prop}{Proposition}
\theoremstyle{definition}
\newtheorem{definition}{Definition}
\usepackage{booktabs}

\usepackage{tikz}
\usepackage{tikz-qtree}
\newcommand{\ra}[1]{\renewcommand{\arraystretch}{#1}}

\usepackage{subcaption}
\usetikzlibrary {graphs,quotes}

\title{Limited Query Graph Connectivity Test}

\author{Mingyu Guo, Jialiang Li, Aneta Neumann, Frank Neumann, Hung Nguyen
\\
\\
School of Computer and Mathematical Sciences, University of Adelaide, Australia\\
\{mingyu.guo, j.li, aneta.neumann, frank.neumann, hung.nguyen\}@adelaide.edu.au\\
}

\begin{document}
\maketitle

\begin{abstract}

We propose a combinatorial optimisation model called {\em Limited Query Graph
    Connectivity Test}.  We consider a graph whose edges have two possible
    states ({\sc On}/{\sc Off}).  The edges' states are hidden initially. We
    could query an edge to reveal its state.  Given a source $s$ and a
    destination $t$, we aim to test $s-t$ connectivity by identifying either a
    {\em path} (consisting of only {\sc On} edges) or a {\em cut} (consisting
    of only {\sc Off} edges). We are limited to $B$ queries, after which we
    stop regardless of whether graph connectivity is established.  We aim to
    design a query policy that minimizes the expected number of queries.

Our model is mainly motivated by a cyber security use case where we need to
    establish whether attack paths exist in a given network, between a source
    (i.e., a compromised user node) and a destination (i.e., a high-privilege
    admin node).  Edge query is resolved by manual effort from the IT admin,
    which is the motivation behind query minimization.

Our model is highly related to {\em Stochastic Boolean Function
    Evaluation (SBFE)}.  There are two existing exact algorithms for SBFE that
    are prohibitively expensive.  We propose a significantly more scalable
    exact algorithm.  While previous exact algorithms only scale for trivial
    graphs (i.e., past works experimented on at most 20 edges), we empirically
    demonstrate that our algorithm is scalable for a wide range of much larger
    practical graphs (i.e., graphs representing Windows domain networks with tens of
    thousands of edges).

    We also propose three heuristics. Our best-performing heuristic is via limiting the planning horizon of the exact algorithm. The other two are via reinforcement learning (RL) and Monte Carlo tree search (MCTS).  We also derive an algorithm for computing the performance lower bound. Experimentally, we show that all our heuristics are near optimal.  The heuristic building on the exact algorithm {\em outperforms all other heuristics}, surpassing RL, MCTS and eight existing heuristics ported from SBFE and related literature.

\end{abstract}

\section{Introduction}

\subsection{Model Motivation}

We propose a model called {\em Limited Query Graph Connectivity Test}, which is mainly motivated by
a cyber security use case.
We start by describing this use case to motivate our model
and also to better explain the model design rationales.
The main focus of this paper is the theoretical model and
the algorithms behind. Nevertheless, our design choices are heavily influenced
by the cyber security use case.


{\em Microsoft Active Directory (AD)} is
the {\em default} security management system for Windows domain networks. An AD
environment is naturally described as a graph where the nodes are
accounts/computers/groups, and the {\em directed} edges represent {\em accesses}. There are many open source and commercial
tools for analysing AD graphs.  {\sc BloodHound}\footnote{https://github.com/BloodHoundAD/BloodHound} is a popular AD analysis tool that is able to enumerate {\em attack paths}
that an attacker can follow through from a source node
to the admin node.
{\sc
ImproHound}\footnote{https://github.com/improsec/ImproHound} is another tool
that is able to flag {\em tier violations} (i.e., this tool
warns if there exist paths that originate from low-privilege nodes
and reach high-privilege nodes).
Existing tools like the above are able to identify attack paths, but they do not provide {\em directly implementable fixes}. In our context, a fix is a set of edges (accesses) that can be {\em safely} removed to eliminate all attack paths.
Unfortunately, this cannot be found via {\em minimum cut}.
Some edges may appear to be redundant but their removal could cause major disruptions (this is like refactoring legacy code -- it takes effort). Consistent with industry practise, we assume that every edge removal is {\em manually} examined and approved, which is therefore costly.  Our vision is a tool  that acts like an ``intelligent wizard'' that guides the IT admin. In every step, the wizard would propose one edge to remove. The IT admin either approves it or rejects it. The wizard acts {\em adaptively} in the sense that
past proposals and their results determine the next edge to propose.
The wizard's goal is to minimize the expected number of proposals. {\em Ultimately, the goal is to save human effort during the algorithm-human collaboration}. We conclude the process if all attack paths have been eliminated, or we have established that elimination is impossible.\footnote{If we cannot
eliminate all attack paths via access removal that are safe, then the
IT admin needs to resort to more ``invasive'' defensive measures, such as
banning accounts or splitting
an account into two --- one for daily usage and one for admin tasks.}
We also conclude if the number of proposals reaches a preset limit.

\subsection{Formal Model Description}

\begin{definition}[{\bf Limited Query Graph Connectivity Test}]
    The {\em Limited Query Graph Connectivity Test} problem involves a graph $G=(V,E)$, either undirected or directed. There is
a single source $s\in V$ and a single destination $t\in V$.
    Every edge $e \in E$ has two possible states ({\sc On}/{\sc Off)}.
    An edge is {\sc On} with probability $p$, where $p$ is a constant model parameter.
    The edges' states are independent.
An edge's state is hidden before we query it.
    By querying an edge, we reveal its state, which will stay revealed and never changes.
    We aim to design an {\em adaptive} query policy that specifies the query order (i.e., past queries and their results decide the next query).
    There are $3$ terminating conditions:

    \begin{itemize}
        \item A {\em path certificate} is found: Among the already revealed edges, the {\sc On} edges form a {\em path} from $s$ to $t$.

        \item A {\em cut certificate} is found: Among the already revealed edges, the {\sc Off} edges form a {\em cut} between $s$ and $t$.

        \item We stop regardless after finishing $B$ queries.
    \end{itemize}

When there isn't ambiguity, we omit the word ``certificate'' and simply use
    {\em path}/{\em cut} to refer to path/cut certificate.

Our optimisation goal is to design a query policy that minimizes the expected number of
    queries spent before reaching termination.
    It should be noted that we are not searching
    for the shortest path or the minimum cut. Any certificate will do.
    If long
    paths/large cuts can be found using less queries, then we prefer them as our goal is solely to minimize the expected query count.

\end{definition}

\begin{figure}
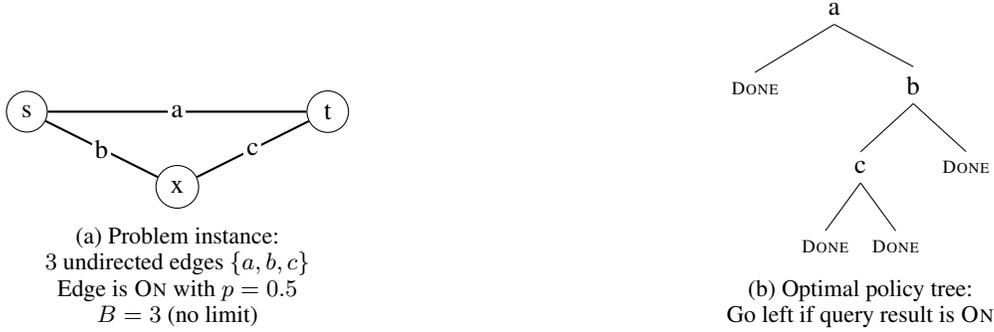

     \centering
     \begin{subfigure}[b]{0.45\textwidth}
\captionsetup{justification=centering}
         \centering
\tikz \graph [edge quotes={fill=white,inner sep=1pt},
          grow down, branch right, nodes={circle,draw}] {
    s --[thick, "a"] t[at={(4,1)}];
    s --[thick, "b"] x[at={(1,-1)}] --[thick, "c"] t
         };
         \caption{Problem instance:\\
         $3$ undirected edges $\{a,b,c\}$\\
         Edge is {\sc On} with $p=0.5$\\
         $B=3$ (no limit)}
         \label{fig:example1left}
     \end{subfigure}
     \hfill
     \begin{subfigure}[b]{0.45\textwidth}
\captionsetup{justification=centering}
         \centering
         \Tree [.a [.{\scriptsize {\sc Done}} ] [.b [.c [.{\scriptsize {\sc Done}} ] [.{\scriptsize {\sc Done}} ] ] [.{\scriptsize {\sc Done}} ] ] ]
         \caption{Optimal policy tree:\\
         Go left if query result is {\sc On}}
         \label{fig:example1right}
     \end{subfigure}
    \caption{{\em Limited Query Graph Connectivity Test} instance}
\label{fig:example1}
\end{figure}

\begin{example} Let us consider the example instance as shown in Figure~\ref{fig:example1left}.
    For this simple instance, the optimal query
    policy is given in Figure~\ref{fig:example1right}, in the form
    of a binary policy tree. The first step is to query edge $a$ in the root position.
    If the query result is {\sc On}, then we go left. Otherwise, we go right.
    The left child of $a$ is {\sc Done}, which indicates that if $a$ is {\sc On},
    then the job is already done ($a$ by itself forms a path).
    The right child of $a$ is $b$, which indicates that if $a$ is {\sc Off},
    then we should query $b$ next. If $b$'s query result is {\sc On}, then
    we query $c$. Otherwise, we go right and reach {\sc Done} (both
    $a$ and $b$ being {\sc Off} form a cut).
    After querying $c$, we reach {\sc Done} regardless of $c$'s state.
    If $c$ is {\sc On}, then we have a path ($b$ and $c$).
    If $c$ is {\sc Off}, then we have a cut ($a$ and $c$).
    The expected number of queries under this optimal policy is
    $1+0.5+0.25=1.75$ (we always query $a$; we query $b$ with $50\%$ chance;
    we query $c$ with $25\%$ chance).
\end{example}

\subsection{Related Models, Key Differences and Rationales}

Our model is highly related to the {\em sequential testing} problem originally from the {\em operation research} community.
An extensive survey can be found in
\cite{Unluyurt04:Sequential}. The sequential testing problem is best
described via the following medical use case. Imagine that a doctor
needs to diagnose a patient for a specific disease. There are $10$ medical tests. Instead of
applying all $10$ tests at once, a cost-saving approach is to test one by one adaptively -- finished tests and their results help pick the next test.
Sequential testing has been applied to iron deficiency anemia diagnosis~\cite{Short13:Iron}. \cite{Yu23:Deep} showed that $85\%$ reduction of medical
cost can be achieved via sequential testing.
In the context of our model, the ``tests'' are the edge queries.

Our model is also highly related to {\em learning with attribute costs} in the context
of {\em machine learning}~\cite{Sun96:Hill, Kaplan05:Learning,
Golovin11:Adaptive}.  Here, the task is to construct
a ``cheap'' classification tree (i.e., shallow decision tree), under the assumption that the
features are costly to obtain.  For example, at the root node, we examine only
one feature and pay its cost.  As we move down the classification tree, every
node in the next/lower layer needs to examine a new feature, which is associated with an additional
cost.  The objective is to
optimize for the cheapest tree by minimizing the expected total feature cost,
weighted by the probabilities of different decision paths, subject to
an accuracy threshold.
In the context of our model, the ``features'' are the edge query results.

The last highly relevant model
is {\em Stochastic Boolean Function Evaluation
(SBFE)}~\cite{Allen17:Evaluation,Deshpande14:Approximation} from
the {\em theoretical computer science} community.
A SBFE instance
involves a Boolean function $f$ with multiple binary inputs and one binary
output.  The input bits are initially hidden and their values follow
independent but not necessarily identical Bernoulli distributions.  Each input
bit has a query cost. The task is to query the input bits in an adaptive order,
until there is enough information to determine the output of $f$. The goal is
to minimize the expected query cost.
In the context of our model, the ``input bits'' are the edges.

All the above models are similar with minor differences on technical details
and they are all highly relevant to our cyber-motivated graph-focused model.
In the appendix, we include more discussion on related research.
Here, we highlight two key differences between our paper and past works.

\vspace{.05in}
\noindent {\bf 1) We focus on {\em empirically scalable exact algorithms}.}

There are two existing exact algorithms from SBFE literature, which are both prohibitively expensive.
\cite{Cox89:Heuristic}'s exact algorithm treats the problem as a MDP and uses
the Bellman equation to calculate the value function for all problem states.
Under our model, with $m$ edges, the number of states is $3^m$ (an edge is hidden, {\sc On}, or {\sc Off}).
\cite{Allen17:Evaluation}
proposed an exact algorithm
with a complexity of $O(n^{2^k})$, where $k$ is either the number of paths or the number
of cuts.
{\em Existing
works on exact algorithms only experimented on tiny instances with at most $20$ edges}~\cite{Fu17:Determining,Bendov81:Branch,Reinwald66:Conversion,Breitbart75:Branch}.
%

Past results mostly focused on heuristics and approximation
algorithms.
\cite{Unluyurt04:Sequential}'s survey 
mentioned
heuristics from~\cite{Jedrzejowicz83:Minimizing,Cox89:Heuristic,Sun96:Hill}.
Approximations algorithms were proposed in~\cite{Allen17:Evaluation,Deshpande14:Approximation,Kaplan05:Learning,Golovin11:Adaptive}.

We argue that
for many applications involving this lineage of models, {\em
algorithm speed is not an important evaluation metric}, which
is why we focus on empirically scalable exact algorithms.
For example, for sequential medical tests,
once a policy tree for diagnosing a specific disease is generated (however slow),
the policy tree can be used for all future patients.
In our experiments, we allocate $72$ hours to the exact algorithm
and we manage to exactly solve
several fairly large graph instances, including one depicting a Windows domain network with $18795$ edges. In comparison, going over $3^{18795}$ states is impossible if we apply
the algorithm from \cite{Cox89:Heuristic}.

Incidentally, while deriving the exact algorithm, we obtain two valuable
by-products: a lower bound algorithm and a heuristic that outperforms all
existing heuristics from literature (we {\em ported} $8$ heuristics to our model).  The exact algorithm works by iteratively generating more
and more expensive policy that eventually converges to the optimal policy.  The
lower bound algorithm is basically settling with the intermediate results.  That
is, if the exact algorithm does not scale (i.e., too slow to converge), then we
stop at any point and whatever we have is a lower bound.  Our best-performing heuristic
also builds on the exact algorithm, which is via limiting the planning horizon
of the exact algorithm.

\vspace{.05in}
\noindent {\bf 2) We introduce a technical concept called {\em query limit}, which enables many positive results.}



{\em The query limit has a strong implication on scalability.}
Our exact algorithm builds on a few technical tricks for scalability, including
a systematic way to identify the relevant edges that may be referenced by the optimal policy
and a systematic way to generate the optimal policy tree structure.
As mentioned earlier, without imposing a query limit,
our algorithm is already capable of optimally solving several fairly large instances, including one with $18795$ edges.
Nevertheless, what really drives up scalability is
the introduction of the query limit. Without the query limit, our exact algorithm
is only scalable for special graphs. On the other hand, as long as the query limit is small, our exact algorithm scales.
It should be noted that this paper is {\bf not} on {\em fixed-parameter analysis} -- our algorithm is {\bf not} {\em fixed-parameter tractable} and the query limit is not the special parameter. The query limit is a technical tool for improving {\bf empirical scalability}.

{\em The query limit is a useful technical tool even for settings without query
limit.} Our lower bound algorithm is more scalable with smaller query limits.
This is particularly nice because for settings without query limit, we could always artificially impose a query limit that is scalable and derive a performance lower bound, as imposing a limit never increases the query count.  We are not aware of any prior work on lower bound for this lineage of models.

Our best-performing heuristic is also based on imposing an artificially small query limit on the exact algorithm (so that it scales -- and we simply follow this limited-horizon exact algorithm).
This way of constructing heuristic should be applicable for settings without query limit.

{\em The query limit is often not a restriction when it comes to graph connectivity test.}
Our experiments demonstrate that for a variety of
                large graphs, the expected query counts required to establish
                graph connectivity are tiny and the query count distributions exhibit clear {\em long-tail} patterns.
                In the appendix, we present histograms of the
                query counts of two large graphs to demonstrate this.
                Essentially, the query limit is not bounding most of the time.

{\em Lastly, the query limit is practically well motivated.} In the context of
our cyber-motivated model, the edge queries are answered by human efforts.
Imposing a query limit is practically helpful to
facilitate this algorithm-human collaboration. We prefer that the IT admin
be able to specify the query limit before launching the interactive session (i.e., ``my time allocation allows at most $10$ queries'').

\subsection{Summary of Results}

\begin{itemize}

    \item
We propose a combinatorial optimisation model called {\em Limited Query
        Graph Connectivity Test}, motivated by a cyber security use case on defending Active Directory managed Windows domain network.

    \item We show that query count minimization is \#P-hard.

    \item
We propose an empirically scalable exact algorithm. It can optimally
        solve practical-scale large graphs when the query limit is small.
        Even when we set the limit to infinity, our algorithm
        is capable of optimally solving several fairly large instances, including
        one with $18795$ edges.
        (Recall that past works on
        exact algorithms only experimented on tiny instances with at most $20$ edges/bits~\cite{Fu17:Determining,Bendov81:Branch,Reinwald66:Conversion,Breitbart75:Branch}.)

    \item Our exact algorithm iteratively generates more and more expensive policy that eventually converges to the optimal policy.
        When the exact algorithm is not scalable, its
        intermediate results serve as performance lower bounds.

    \item We propose three heuristics. Our best-performing heuristic
is via imposing an artificially small query limit on the exact algorithm.
The other two heuristics are based on reinforcement learning (RL) and Monte Carlo tree search (MCTS),
where the action space is reduced with the help of query limit.
We experiment on a
        wide range of practical graphs, including road and power networks,
        Python package dependency graphs and Microsoft Active Directory
        graphs.
        We conduct a comprehensive survey on existing heuristics
        and approximation algorithms on sequential testing, learning
        with attribute costs and stochastic Boolean function evaluation.
        Our heuristic building on the exact algorithm outperforms all, surpassing RL, MCTS and $8$ heuristics ported from literature.

    \item Our techniques have the potential to be applicable to other models related to sequential
testing, learning with attribute costs and stochastic Boolean function
evaluation.  The high-level idea of using query limit
to derive empirically scalable exact algorithm is general.
Using the query limit as the
technical tool to derive high-quality heuristic and performance lower bound is also
general, and it is applicable to settings without query limit.

\end{itemize}

\section{Scalable Exact Algorithm}

We preface our exact algorithm with a hardness result.

\begin{thm}
    \label{thm:sharpp}
    For the Limited Query Graph Connectivity Test problem, it is \#P-hard to compute the minimum expected number of queries.
\end{thm}

The proof is deferred to the appendix.


\subsection{Outer Loop: Identifying Paths/Cuts Worthy of Consideration and Exactness Guarantee}

We first
propose a systematic way for figuring out which paths/cuts are relevant to the
optimal query policy.
A path/cut is relevant if it is referenced by the query policy -- the policy
terminates after establishing the path/cut. In other words, the policy
tree's leaf nodes correspond to the path/cut.
Generally speaking, when the query limit
is small, only a few paths/cuts are relevant.
As an extreme example, if the query limit is $2$, then at most there are $4$ leaf nodes in the optimal policy tree, which
correspond to a total of $4$ paths/cuts ever referenced.
The challenge
is to identify, continuing on the above example, which $4$ are used among the {\em exponential} number of paths/cuts.
Restricting attention to relevant paths/cuts will greatly reduce the
search space, as we only need to focus on querying the edges that are part of them.

\vspace{.05in}
\noindent{\em Key observation:}
Our terminating condition is to find a path, or find a cut, or reach
the query limit.
Finding a path can be reinterpreted as {\em disproving all cuts}.
Similarly, finding a cut can be reinterpreted as {\em disproving all paths}.
\vspace{.05in}

Suppose we have a path set $P$ and a cut set $C$, we introduce a {\em cheaper}
({\em cheaper or equal}, to be more accurate)
problem, which is to come up with an optimal query policy that disproves all
paths in $P$, or disproves all cuts in $C$, or reach the query limit. We use
$\pi(P, C)$ to denote the optimal policy for this cheaper task, and use $c(P, C)$
to denote the minimum number of queries in expectation under $\pi(P,C)$.
We use $\pi^*$ to denote the optimal policy for the original problem,
and use $c^*$ to denote the minimum number of queries in expectation
under $\pi^*$.

\begin{prop}
    \label{prop:cheaper}
    Given $P\subset P'$ and $C\subset C'$, we must have $c(P, C) \le c(P', C')$,
    which also implies $c(P, C) \le c^*$.
\end{prop}

\begin{proof}
The optimal policy $\pi(P', C')$ can still be applied
to the instance with less paths/cuts to disprove. $\pi(P',C')$ has the potential to terminate even earlier when there are less to disprove.
Lastly, $c^*$ is just $c(\text{all paths},\text{all cuts})$.
\end{proof}

Suppose we have access to a subroutine that calculates $\pi(P,C)$ and $c(P,C)$,
which will be introduced soon in the remaining of this section.
Proposition~\ref{prop:cheaper} leads to the following iterative exact algorithm.
Our notation follow Figure~\ref{fig:example1}.

\begin{enumerate}

    \item We initialize an arbitrary path set $P$ and an arbitrary cut set $C$. In our experiments,
        we simply start with singleton sets (one shortest path and one minimum cut).

    \item Compute $\pi(P,C)$ and $c(P,C)$; $c(P,C)$ becomes the best lower bound
        on $c^*$ so far.

    \item Go over all the leaf nodes of the policy tree behind $\pi(P,C)$ that
        terminate due to successfully disproving either $P$ or $C$.
        Suppose we are dealing with a node $x$ that has disproved all paths in $P$.
        At node $x$, we have the results of several queries (i.e., when going
        from the policy tree's root node to node $x$, any left turn corresponds to an {\sc On} edge, and any
        right turn corresponds to an {\sc Off} edge).
        We check whether the confirmed {\sc Off} edges at node $x$ form a cut.
        If so, then we have not just disproved all paths {\bf in} $P$,
        but also disproved all paths {\bf outside} $P$ as well.
        If the confirmed {\sc Off} edges do not form a cut, then that means
        we can find a path that has yet been disproved.
        We add this path to $P$ --- our policy made a mistake (prematurely declared {\sc Done}) because we have
        not considered this newly identified path.

        We add new cuts to $C$ using the same method.

    \item
    If new paths/cuts were added, then go back to step $2$ and repeat.
        If no new paths/cuts were added, then that means when $\pi(P,C)$
        terminates, it has found a path, or found a cut, or reached the query limit. So $\pi(P,C)$ is a feasible
        query policy for the original problem. Proposition~\ref{prop:cheaper} shows
        that $c(P,C)$ never exceeds $c^*$, so $\pi(P,C)$ must be an optimal policy
        for the original problem.
\end{enumerate}

In the appendix, we show the running details of our exact algorithm
on a graph called {\sc euroroad} (a road network with $1417$ undirected edges~\cite{konect}).
For $B=10$, when the exact algorithm terminates, the final path set size
is only $39$ and the final cut set size is only $18$. That is, our proposed approach successfully
picked out only a small number of relevant paths/cuts.

\subsection{Inner Loop: Designing Policy Tree Structure}

We now describe the subroutine that computes $\pi(P,C)$ and $c(P,C)$ given the
path set $P$ and the cut set $C$ as input.

This subroutine is further divided into two subroutines.  First, we need to come up with an optimal {\em tree structure} (i.e., the shape of the policy tree).
After that, we need to decide how to optimally place the queries into the tree structure.

We first focus on finding the optimal tree structure.
Every policy tree is binary,
but it is generally not going to be a
complete tree of $B$ layers (unless we hit the query limit $100\%$ of the time). For example, for the optimal policy tree in
Figure~\ref{fig:example1right},
the left branch has only one node,
while the right branch is deeper and more complex.
The reason we care about the tree structure is that a small tree involves less decisions, therefore easier to design than a complete
tree.


Given $P$ and $C$, we define a policy tree to be {\em correct} if all the {\sc Done}
nodes are labelled correctly (i.e., it is indeed done when the policy tree claims done). A correct policy tree does not have to reach conclusion in every leaf node (we allow a chain of queries to end up inconclusively). We also allow {\sc Done} nodes to have children, but they must all be {\sc Done}.
We use the instance from Figure~\ref{fig:example1left} as an example.
Suppose $P$ is the set of all paths and $C$ is the set of all cuts.
Figure~\ref{fig:correct} is a correct policy tree, because when
it claims {\sc Done}, $a$ is {\sc On}, which correctly disproves every cut in $C$.
As you can see, we allow the right branch to contain a partial policy that is
inconclusive (after querying $b$, the remaining policy is not specified).
Figure~\ref{fig:incorrect} is an incorrect policy tree because the {\sc Done}
node
on the right-hand side is a wrong claim. $a$ being {\sc Off} is not enough
to disprove every path in $P$.

\begin{figure}[h]
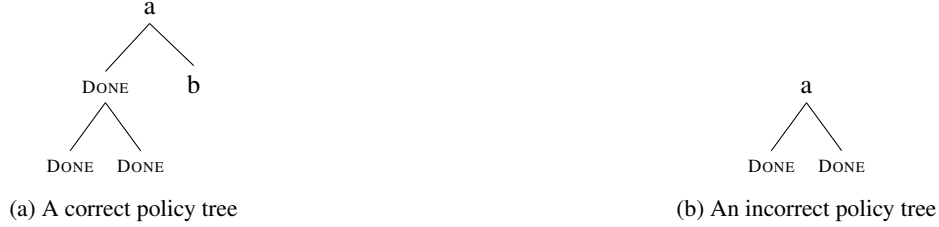

\captionsetup{justification=centering}
     \centering
     \begin{subfigure}[b]{0.45\textwidth}
         \centering
         \Tree [.a [.{\scriptsize {\sc Done}} [.{\scriptsize {\sc Done}} ] [.{\scriptsize {\sc Done}} ] ] [.b ] ]

         \caption{A correct policy tree}
         \label{fig:correct}
     \end{subfigure}
     \hfill
     \begin{subfigure}[b]{0.45\textwidth}
         \centering
         \Tree [.a [.{\scriptsize {\sc Done}} ] [.{\scriptsize {\sc Done}} ] ]
         \caption{An incorrect policy tree}
         \label{fig:incorrect}
     \end{subfigure}
    \caption{We consider the instance in Figure~\ref{fig:example1left}.\\
    $P=\{(a),(b,c)\}$ and $C=\{(a,b),(a,c)\}$}
\label{fig:example2}
\end{figure}

Any policy tree for the original problem must be a correct policy tree.
The cost of a correct policy tree has the following natural definition. {\sc Done} nodes
have no costs. Every query node costs $1$ and is weighted according
to the probability of reaching the node.
Given a tree structure $S$, we define $T(S,P,C)$ to be the minimum
cost correct policy tree, as the result of optimally ``filling in'' queries and {\sc Done} into the tree structure $S$, subject
to the only constraint that the resulting tree must be correct with respect to $P$ and $C$.
We define $c(S,P,C)$ to be the cost of $T(S, P, C)$.

\begin{prop} For any tree structure $S$, path set $P$ and cut set $C$, we have $c(S, P, C) \le c(P, C)$.
    \label{prop:cheaper2}
\end{prop}

\begin{proof}
The policy tree $T^*$ behind $c(P,C)$ can be imposed onto the tree structure $S$
and results in a correct policy tree. For example, Figure~\ref{fig:correct} is
a correct policy tree as a result of imposing Figure~\ref{fig:example1right} to
this new structure. Basically, any node position in $S$ that is not in $T^*$ can be
    filled in with {\sc Done} (since $T^*$'s leave must be {\sc Done} or has already reached the depth limit). Adding {\sc Done} does not incur any cost. Some query
nodes from $T^*$ may be dropped if their positions are not in $S$, which reduces cost.
So $T^*$ can be converted to a correct policy tree with structure $S$ that is not more expensive.
\end{proof}

\begin{prop}
    \label{prop:cheaper3}
    If $S$ is a subtree of $S'$, then $c(S, P, C) \le c(S', P, C)$.
\end{prop}

\begin{proof}
$T(S',P,C)$ can be imposed onto $S$. The resulting tree is still correct and never more expensive.
\end{proof}

We will discuss how to calculate $T(S,P,C)$ and $c(S,P,C)$
(i.e., how to optimally fill in a given tree structure)
toward the end of this section. Assuming access to $T(S,P,C)$ and $c(S,P,C)$, Proposition~\ref{prop:cheaper2} and~\ref{prop:cheaper3} combined lead to the following iterative algorithm for designing the optimal tree structure:

\begin{enumerate}

    \item We initialize an arbitrary tree structure $S$. In experiments,
        we use a complete binary tree with $4$ layers.

    \item $P$ and $C$ are given by the algorithm's outer loop. We calculate $c(S,P,C)$, which becomes the best lower bound on $c(P,C)$ so far.

    \item Go over all the leaf nodes of $T(S,P,C)$.
        If there is a leaf node that is not conclusive (not {\sc Done} and not  reaching the query limit), then we expand this node by attaching two child nodes to the tree structure. That is, in the next iteration, we need to decide what queries to place into these two new slots. The cost for the next iteration never decreases according to Proposition~\ref{prop:cheaper3}.

    \item
    If the tree was expanded, then go back to step $2$ and repeat.
        If no new nodes were added, then that means $T(S,P,C)$
        involves no partial decision. Its cost $c(S,P,C)$ must be at least $c(P,C)$ as $c(P,C)$
        is supposed to be the minimum cost for disproving either $P$ or $C$ (or reach the query limit). This combined with Proposition~\ref{prop:cheaper2} imply
        that $c(S,P,C)=c(P,C)$.
\end{enumerate}

Earlier, we mentioned our algorithm's running details
on a graph called {\sc euroroad} with $B=10$.
For this instance, under our exact algorithm, the final tree structure size is
$91$.
On the
contrary, if we do not use the above ``iterative tree growth'' idea and simply work on a complete
tree, then the number of slots is $2^{10}-1=1023$.
Our approach managed to significantly
reduce the search space.

In the above algorithm description, we separated the ``outer'' and ``inner'' loops.
This is purely for cleaner presentation. In implementation, essentially what we do is to iteratively generate the best policy tree so far given the current path set $P$, the current cut set $C$ and the current tree structure $S$.
If the resulting policy tree makes any wrong claims (claiming {\sc Done} prematurely), then we
expand the path/cut set. If the tree goes into any inconclusive situation (reaching
the leave position and still cannot claim {\sc Done}), then
we expand the tree structure. During this iterative process, we keep getting equal or higher cost based on Proposition~\ref{prop:cheaper} and \ref{prop:cheaper3}. When the process converges, which is {\em theoretically} guaranteed ({\em not practically}), we have the optimal query policy.
The exact algorithm can be interrupted anytime and the intermediate results serve as performance lower bounds.


\subsection{Optimal Correct Tree via Integer Program}

We now describe the last subroutine. Given a path set $P$, a cut set $C$ and a tree structure $S$, we aim to build the minimum cost correct policy tree $T(S,P,C)$.

Our technique is inspired by existing works on building optimal decision trees
for machine learning classification tasks using integer
programming~\cite{Bertsimas17:Optimal,Verwer19:Learning}.
The works on optimal decision trees studied how to select the features to test
at the tree nodes.
The learning samples are routed down the tree based on the selected features
and the feature test results. The objective is to maximize the
overall accuracy at the leaf nodes.  We are performing
a very similar task.
The objective is no longer about learning accuracy but on minimizing tree cost in the context
of our model. Our model is as follows.

We first present the {\bf variables}. Let $E^R$ be the set of edges referenced in either $P$ or $C$.
For every tree node $i\in S$,
for every edge $e\in E^R$,
we define a binary variable $v_{e,i}$. If $v_{e,i}$ is $1$, then it means
we will perform query $e$ at node $i$.
We introduce another variable $v_{\textsc{Done},i}$, which is $1$ if and only if
it is correct to claim {\sc Done} at node $i$.

Next, we present {\bf constraints}.

Every node is either a query node or {\sc Done}, which implies
$\forall i, \sum_{e\in E^R}v_{e,i}+v_{\textsc{Done},i} =1$.

If a node is {\sc Done}, then all its descendants should be {\sc Done}, which implies
$\forall i, v_{\textsc{Done},i}\ge v_{\textsc{Done},\textsc{Parent}(i)}$.

Along the route from root to any leaf, each edge should be queried at most once.
Let $\textsc{route}(i)$ be the set of nodes along the route from root to node $i$ (inclusive). We have
$\forall i\in \textsc{Leaves},\forall e\in E^R, \sum_{j\in \textsc{route}(i)} v_{e,j}\le 1$.

{\sc Done} claims must be correct.
Let $\textsc{lnodes}$ be the set of tree nodes who are left children of their parents.
If we claim {\sc Done} at node $i\in\textsc{lnodes}$ and $\textsc{Parent}(i)$ is not already {\sc Done}, then we must have disproved all cuts in $C$ once reaching $i$, since $\textsc{Parent}(i)$'s query result
is {\sc On}. To verify that we have disproved all cuts in $C$, we only need to consider the {\sc On} edges confirmed once reaching $i$.
The queries that resulted in {\sc On} edges are queried at the following nodes
$N(i)=\{\textsc{Parent}(j)|j \in \textsc{lnodes}\cap \textsc{route}(i)\}$.
Given a $\textsc{cut}$ (interpreted as a set of edges), to disprove it, all we need is that at least one edge in the cut has been queried and has returned an {\sc On} result, which is $\sum_{j\in N(i),e\in\textsc{cut}}v_{e,j} \ge 1$.

We only need to verify that {\sc Done} is correctly claimed when it is claimed
for the first time, which would be for nodes who satisfy that $v_{\textsc{Done},i}-v_{\textsc{Done},\textsc{Parent}(i)}=1$.
In combination, for node $i\in\textsc{lnodes}$, we have the following constraint: $\forall\textsc{cut}\in C, \sum_{j\in N(i),e\in\textsc{cut}}v_{e,j}-(v_{\textsc{Done},i}-v_{\textsc{Done},\textsc{Parent}(i)}) \ge 0$.
Note that
$v_{\textsc{Done},i}-v_{\textsc{Done},\textsc{Parent}(i)}$ is $0$ when we
are not dealing with a node that claims {\sc Done} for the first time, which
would disable the above constraint as it is automatically satisfied.

The above constraints only referenced cuts. We construct constraints in a similar manner for disproving paths.

The {\bf objective} is to minimize the tree cost. The cost of node $i$ is simply
$\sum_{e\in E^R}v_{e,i}$. The probability of reaching a node only depends
on the input tree structure so it is a constant.
We use $\textsc{Prob}(i)$ to denote the probability
of reaching node $i$. For example, if it takes $3$ left turns ($3$ {\sc On} results) and $2$ right turns (2 {\sc Off} results) to reach $i$, then
$\textsc{Prob}(i)=p^3(1-p)^2$.

In the appendix,
we provide the complete integer program, the overall algorithm's pseudocode and a detailed walk through of our algorithm on an example instance.


\begin{table*}[t!]\centering

\ra{1}

\footnotesize

\begin{tabular}{@{}rrrcrrcrrrrrr@{}}\toprule

& \multicolumn{2}{c}{Size} & \phantom{abc} & \multicolumn{2}{c}{$B=5$} & \phantom{abc} & \multicolumn{5}{c}{$B=10$}\\

\cmidrule{2-3} \cmidrule{5-6} \cmidrule{8-12}

    & $n$ & $m$ & & Exact & Other & & Exact/LB & Tree & RL & MCTS & Other \\

\midrule

\addlinespace[.01in]\multicolumn{10}{l}{{\em Road networks - Undirected}}\\\addlinespace[.03in]
\multicolumn{1}{l}{\sc{USAir97}} &333 &2126 && ${\bf3.813}$ & ${\bf3.813}$ && 4.555 & ${\bf4.986}$ & 5.029 & 4.992 & 4.994 \\
\multicolumn{1}{l}{\sc{euroroad}} &1175 &1417 && ${\bf1.938}$ & ${\bf1.938}$ && ${\bf2.109}$ & ${\bf2.109}$ & ${\bf2.109}$ & ${\bf2.109}$ & ${\bf2.109}$ \\
\multicolumn{1}{l}{\sc{minnesota}} &2643 &3303 && ${\bf3.938}$ & ${\bf3.938}$ && 4.785 & 5.506 & 5.525 & 5.506 & ${\bf\underline{5.477}}$ \\
\addlinespace[.04in]\multicolumn{10}{l}{{\em Power networks - Undirected}}\\\addlinespace[.03in]
\multicolumn{1}{l}{\sc{power}} &4942 &6594 && ${\bf3.625}$ & ${\bf3.625}$ && 4.367 & ${\bf4.615}$ & 4.617 & ${\bf4.615}$ & ${\bf4.615}$ \\
\multicolumn{1}{l}{\sc{bcspwr01}} &40 &85 && ${\bf3.125}$ & ${\bf3.125}$ && ${\bf3.541}$ & ${\bf3.541}$ & 3.547 & ${\bf3.541}$ & ${\bf3.541}$ \\
\multicolumn{1}{l}{\sc{bcspwr02}} &50 &108 && ${\bf2.625}$ & ${\bf2.625}$ && ${\bf3.178}$ & ${\bf3.178}$ & 3.207 & 3.229 & ${\bf3.178}$ \\
\multicolumn{1}{l}{\sc{bcspwr03}} &119 &297 && ${\bf3.938}$ & ${\bf3.938}$ && 4.883 & ${\bf5.816}$ & 5.895 & 5.887 & 5.912 \\
\multicolumn{1}{l}{\sc{bcspwr04}} &275 &943 && ${\bf3.938}$ & ${\bf3.938}$ && 5.188 & ${\bf6.074}$ & 6.096 & 6.086 & 6.080 \\
\multicolumn{1}{l}{\sc{bcspwr05}} &444 &1033 && ${\bf2.938}$ & ${\bf2.938}$ && 3.555 & 4.551 & 4.543 & 4.561 & ${\bf4.539}$ \\
\multicolumn{1}{l}{\sc{bcspwr06}} &1455 &3377 && ${\bf2.625}$ & ${\bf2.625}$ && 3.164 & ${\bf3.197}$ & 3.211 & 3.246 & 3.201 \\
\multicolumn{1}{l}{\sc{bcspwr07}} &1613 &3718 && ${\bf2.625}$ & 2.689 && 3.346 & ${\bf3.461}$ & 3.500 & 3.564 & 3.625 \\
\multicolumn{1}{l}{\sc{bcspwr08}} &1625 &3837 && ${\bf4.188}$ & ${\bf4.188}$ && 4.945 & ${\bf5.996}$ & 6.174 & 6.266 & 6.414 \\
\multicolumn{1}{l}{\sc{bcspwr09}} &1724 &4117 && ${\bf3.938}$ & ${\bf3.938}$ && 4.664 & ${\bf5.891}$ & 5.895 & 5.924 & ${\bf5.891}$ \\
\multicolumn{1}{l}{\sc{bcspwr10}} &5301 &13571 && ${\bf4.625}$ & ${\bf4.625}$ && 5.270 & 7.428 & 7.393 & ${\bf7.350}$ & 7.553 \\
\addlinespace[.04in]\multicolumn{10}{l}{{\em Miscellaneous small graphs - Directed except for Chesapeake}}\\\addlinespace[.03in]
\multicolumn{1}{l}{\sc{bugfix}} &13 &28 && ${\bf2.438}$ & ${\bf2.438}$ && ${\bf2.715}$ & ${\bf2.715}$ & ${\bf2.715}$ & ${\bf2.715}$ & ${\bf2.715}$ \\
\multicolumn{1}{l}{\sc{football}} &36 &118 && ${\bf1.750}$ & ${\bf1.750}$ && ${\bf1.750}$ & ${\bf1.750}$ & ${\bf1.750}$ & ${\bf1.750}$ & ${\bf1.750}$ \\
\multicolumn{1}{l}{\sc{chesapeake}} &40 &170 && ${\bf3.813}$ & ${\bf3.813}$ && 4.594 & ${\bf5.090}$ & 5.152 & 5.199 & 5.113 \\
\multicolumn{1}{l}{\sc{cattle}} &29 &217 && ${\bf2.875}$ & ${\bf2.875}$ && 3.721 & ${\bf4.131}$ & 4.166 & 4.158 & \underline{4.188} \\
\addlinespace[.04in]\multicolumn{10}{l}{{\em Python dependency graphs - Directed}}\\\addlinespace[.03in]
\multicolumn{1}{l}{\sc{networkx}} &423 &908 && ${\bf2.000}$ & ${\bf2.000}$ && ${\bf2.078}$ & ${\bf2.078}$ & ${\bf2.078}$ & ${\bf2.078}$ & ${\bf2.078}$ \\
\multicolumn{1}{l}{\sc{numpy}} &429 &1208 && ${\bf2.688}$ & ${\bf2.688}$ && 3.418 & ${\bf3.688}$ & 3.717 & 3.709 & 3.719 \\
\multicolumn{1}{l}{\sc{matplotlib}} &282 &1484 && ${\bf2.188}$ & ${\bf2.188}$ && ${\bf2.520}$ & ${\bf2.520}$ & ${\bf2.520}$ & 2.539 & \underline{2.537} \\
\addlinespace[.04in]\multicolumn{10}{l}{{\em Active Directory graphs - Directed}}\\\addlinespace[.03in]
\multicolumn{1}{l}{\sc{r2000}} &5997 &18795 && ${\bf1.938}$ & ${\bf1.938}$ && ${\bf1.938}$ & ${\bf1.938}$ & ${\bf1.938}$ & ${\bf1.938}$ & ${\bf1.938}$ \\
\multicolumn{1}{l}{\sc{r4000}} &12001 &45781 && ${\bf2.688}$ & ${\bf2.688}$ && 3.268 & ${\bf3.344}$ & 3.352 & ${\bf3.344}$ & ${\bf3.344}$ \\
\multicolumn{1}{l}{\sc{ads5}} &1523 &5359 && ${\bf2.688}$ & ${\bf2.688}$ && 3.287 & ${\bf3.391}$ & 3.402 & ${\bf3.391}$ & ${\bf3.391}$ \\
\multicolumn{1}{l}{\sc{ads10}} &3015 &12776 && ${\bf3.375}$ & ${\bf3.375}$ && 3.957 & ${\bf4.123}$ & ${\bf4.123}$ & ${\bf4.123}$ & ${\bf4.123}$ \\
\bottomrule

\end{tabular}

\caption{Expected query count under different algorithms: {\sc Exact}=Exact algorithm; {\sc Exact/LB}=Either exact (bold) or lower bound; {\sc Other}=Best out of 8 heuristics from literature (those underlined are {\bf not} achieved by {\sc H1}); {\sc Tree}=Heuristic based on the exact algorithm; {\sc RL} and {\sc MCTS} are self-explanatory}

\label{tablebig}

\end{table*}

\section{Heuristics}
There are a long list of existing heuristics and approximation
algorithms from literature on sequential testing, learning with attribute costs and stochastic
Boolean function evaluation.
\cite{Unluyurt04:Sequential}'s survey mentioned
heuristics from~\cite{Jedrzejowicz83:Minimizing,Cox89:Heuristic,Sun96:Hill}.
Approximations algorithms were proposed in~\cite{Allen17:Evaluation,Deshpande14:Approximation,Kaplan05:Learning,Golovin11:Adaptive}.
We {\em select} and {\em port} $8$ heuristics/approximation algorithms from literature for comparison.
In the appendix, we describe these algorithms and provide details on the necessary modifications we applied
to fit our model.
Below, we present one heuristic, referred to as {\sc H1}. {\sc H1} is {\em impressively elegant} and experimentally it is the {\em dominantly better} existing heuristic.
The definition is worded in the context of our graph
model.



\begin{definition}[{\sc H1}~\cite{Jedrzejowicz83:Minimizing}, also independently discovered in \cite{Cox89:Heuristic}]
    We find the path with the minimum number of hidden edges from $s$ to $t$.
    The path is not allowed to contain {\sc Off} edges.
    We also find the cut with the minimum number of hidden edges between $s$ to $t$.
    The cut is not allowed to contain {\sc On} edges.
    The path and the cut must intersect and the intersection must be a hidden edge. We query this edge.
    \label{def:h1}
\end{definition}


Below we present three new heuristics for our model.
All hyper-parameters are deferred to the appendix.

\noindent {\bf Heuristic based on the exact algorithm:}
    We pretend that there are only $B'$ queries left ($B'<B$) and apply the exact algorithm.
   $B'$ needs to be small enough so that the exact algorithm is scalable.
    Once we have the optimal policy tree, we query the edge specified in the tree root, and then re-generate the optimal policy tree for the updated graph, still pretending that there are only $B'$ queries left.
We refer to this heuristic as {\sc Tree}. 


\noindent {\bf Reinforcement learning:}
\cite{Yu23:Deep} applied reinforcement learning to sequential medical tests involving
$6$ tests. We cannot directly apply reinforcement learning to our model, as we
have a lot more than $6$ queries to arrange.
For our case, the action space is huge, as large graphs contain tens of thousands of edges.
The key of our algorithm is a heuristic for limiting the action space.
We use {\sc H1} to heuristically generate a small set
of top priority edges to select from. Specifically, we generate {\sc H1}'s policy
tree assuming a query limit of $B'$. All edges referenced form the action space.
The action space size is at most $2^{B'}-1$.

It is also not scalable to take the whole graph as {\em observation}.
Fortunately, {\em the query limit comes to assist}.  Our observation contains $B-1$
segments as there are at most $B-1$ past queries. Each segment contains $3+(2^{B'}-1)$ bits. $3$ bits are one hot encoding of query result (not queried yet, {\sc On}, {\sc Off}). $2^{B'}-1$
bits encode the action.

Since the action space always contains {\sc H1}'s default action, in experiments,
when we apply Proximal Policy Optimization (PPO)~\cite{Schulman17:Proximal}, the agent ended up cloning
{\sc H1} and we cannot derive better/new heuristics.
We instead apply Soft Actor-Critic for Discrete Action~\cite{Christodoulou19:Soft}, as it
involves entropy maximisation, which encourages the agent to deviate from {\sc H1}.
We managed to obtain better policy than {\sc H1} on many occasions.

\noindent {\bf Monte Carlo tree search:}
We use the same heuristic to limit the action space (same as {\sc RL}).
We apply standard epsilon-greedy Monte Carlo tree search.

\section{Experiments}

We experiment on a wide range of practical graphs, including
road and power networks~\cite{Davis11:University,konect,Rossi15:Network,Watts98:Collective,bcspwr}, Python package dependency graphs~\cite{pydeps}, Microsoft Active Directory attack graphs~\cite{dbcreator,adsimulator,Guo22:Practical,Goel22:Defending,Guo23:Scalable,Goel23:Evolving} and more~\cite{Allard19:Dataset}. For all experiments, we set $p=0.5$. Sources and destinations are selected
randomly. Please refer to the appendix for reproducibility nodes and also the graph descriptions.

Our results are summarized in Table~\ref{tablebig}.
{\em When $B=5$, our exact algorithm scales for all $25$ graphs.}
The column {\sc Other} shows the best performance among $8$ heuristics
from literature.
Besides {\sc H1} from Definition~\ref{def:h1},
the other $7$ heuristics are described in the appendix.
For $B=5$, {\sc H1} performs the best among all existing heuristics for every single graph.
Actually, {\sc H1} produces
the optimal results (confirmed by our exact algorithm) in all but one graph ({\sc bcspwr07}).

{\em When $B=10$, our exact algorithm produces the optimal policy for $8$ out of $25$ graphs} (indicated by bold numbers in column {\sc Exact/LB}). For the remaining graphs, the exact algorithm only produces a lower bound (we set a time out of $72$ hours). {\em The achieved lower bounds have high quality} (i.e., close to the achievable results).

When $B=10$, we compare our best-performing heuristic {\sc Tree} against the
best existing heuristic (best among $8$).  {\sc Tree} wins for $11$ graphs,
loses for $2$ graphs, and ties for $12$ graphs.  It should be noted that in the
above, we are comparing {\sc Tree} against the {\bf best} existing heuristic
(best among $8$). If we are comparing against any {\bf individual} heuristic,
then our ``winning rate'' would be even higher.


{\sc RL} and {\sc MCTS} perform worse than {\sc Tree}.
Nevertheless, if we take the better between {\sc RL} and {\sc MCTS}, the results outperform the existing heuristics (best among $8$).
We suspect that the ``long-tail'' pattern of the query costs causes difficulty for both {\sc RL} and {\sc MCTS}. The real advantage of {\sc RL} and {\sc MCTS} is that they are model-free.
In the appendix, we discuss model generalisation involving interdependent edges,
which can easily be handled by {\sc RL} and {\sc MCTS}.




{\em Lastly, our exact algorithm can find the optimal solution for
{\sc bugfix}, {\sc r2000}, {\sc football} and {\sc networkx} even without query limit}.
Among them, {\sc r2000} is the largest with $18795$ edges.

To summarize, {\bf our exact algorithm can handle significantly larger
graphs} compared to experiments from existing literature using at
most $20$ edges~\cite{Bendov81:Branch,Reinwald66:Conversion,Breitbart75:Branch}.
Even when our exact algorithm does not terminate, it produces
a {\bf high-quality lower bound}. Our heuristics, especially {\sc Tree},
{\bf outperforms all existing heuristics} from literature.


\newpage
\bibliographystyle{abbrv}
\bibliography{/home/mingyu/Dropbox/nixos/extra,/home/mingyu/Dropbox/nixos/mg,/home/mingyu/Dropbox/nixos/mggame,/home/mingyu/Dropbox/nixos/mingyu_dblp,extra}

\newpage

\section{Acknowledgments}
Frank Neumann has been supported by the Australian Research Council (ARC) through grant FT200100536.

\section*{Appendix}

\subsection{Experiments for $B=20$}

Table~\ref{tablebig20} shows the experimental results for $B=20$.
We compare our heuristics (best among {\sc Tree}, {\sc RL} and {\sc MCTS}) against the existing heuristics (best among $8$ from literature).
There are $25$ graphs in total.
Our heuristics come out on top for $15$ graphs, loses for $6$ graphs and ties for $4$ graphs.
If we only compare our best-performing heuristic {\sc Tree} against the best existing heuristic,
then {\sc Tree} wins for $13$ graphs, loses for $6$ graphs, and ties for $6$ graphs.

\begin{table*}[t!]\centering

\ra{1}

\footnotesize

\begin{tabular}{@{}rrrcrrrr@{}}\toprule

& \multicolumn{2}{c}{Size} & \phantom{abc} & \multicolumn{4}{c}{$B=20$}\\

\cmidrule{2-3} \cmidrule{5-8}

    & $n$ & $m$ & & Tree & RL & MCTS & Other \\

\midrule

\multicolumn{7}{l}{{\em Road networks - Undirected}}\\
\multicolumn{1}{l}{\sc{USAir97}} &333 &2126 && ${\bf5.484}$ & 5.660 & 5.525 & 5.591 \\
\multicolumn{1}{l}{\sc{euroroad}} &1175 &1417 && ${\bf2.121}$ & ${\bf2.121}$ & ${\bf2.121}$ & ${\bf2.121}$ \\
\multicolumn{1}{l}{\sc{minnesota}} &2643 &3303 && ${\bf6.504}$ & 6.677 & ${\bf6.504}$ & 6.518 \\
\multicolumn{7}{l}{{\em Power networks - Undirected}}\\
\multicolumn{1}{l}{\sc{power}} &4942 &6594 && 5.118 & 5.134 & 5.210 & ${\bf5.111}$ \\
\multicolumn{1}{l}{\sc{bcspwr01}} &40 &85 && 3.627 & 3.627 & 3.627 & ${\bf\underline{3.625}}$ \\
\multicolumn{1}{l}{\sc{bcspwr02}} &50 &108 && 3.149 & 3.156 & ${\bf3.130}$ & 3.149 \\
\multicolumn{1}{l}{\sc{bcspwr03}} &119 &297 && ${\bf7.400}$ & 7.773 & 7.833 & \underline{7.748} \\
\multicolumn{1}{l}{\sc{bcspwr04}} &275 &943 && ${\bf7.786}$ & 7.993 & 7.927 & 7.853 \\
\multicolumn{1}{l}{\sc{bcspwr05}} &444 &1033 && 6.513 & 6.706 & 6.568 & ${\bf6.441}$ \\
\multicolumn{1}{l}{\sc{bcspwr06}} &1455 &3377 && 3.277 & 3.261 & 3.399 & ${\bf3.260}$ \\
\multicolumn{1}{l}{\sc{bcspwr07}} &1613 &3718 && 3.908 & 4.044 & ${\bf3.864}$ & 4.020 \\
\multicolumn{1}{l}{\sc{bcspwr08}} &1625 &3837 && ${\bf7.253}$ & 8.038 & 7.448 & 7.675 \\
\multicolumn{1}{l}{\sc{bcspwr09}} &1724 &4117 && 8.284 & 8.606 & 9.018 & ${\bf8.203}$ \\
\multicolumn{1}{l}{\sc{bcspwr10}} &5301 &13571 && 10.672 & 11.189 & 12.066 & ${\bf10.666}$ \\
\multicolumn{7}{l}{{\em Miscellaneous small graphs - Directed except for Chesapeake}}\\
\multicolumn{1}{l}{\sc{bugfix}} &13 &28 && 2.708 & ${\bf2.699}$ & 2.711 & \underline{2.708} \\
\multicolumn{1}{l}{\sc{football}} &36 &118 && ${\bf1.777}$ & ${\bf1.777}$ & ${\bf1.777}$ & ${\bf1.777}$ \\
\multicolumn{1}{l}{\sc{chesapeake}} &40 &170 && ${\bf5.507}$ & 5.727 & 5.574 & 5.615 \\
\multicolumn{1}{l}{\sc{cattle}} &29 &217 && ${\bf4.507}$ & 4.619 & 4.544 & 4.604 \\
\multicolumn{7}{l}{{\em Python dependency graphs - Directed}}\\
\multicolumn{1}{l}{\sc{networkx}} &423 &908 && ${\bf2.058}$ & ${\bf2.058}$ & ${\bf2.058}$ & ${\bf2.058}$ \\
\multicolumn{1}{l}{\sc{numpy}} &429 &1208 && ${\bf4.239}$ & 4.339 & 4.310 & 4.311 \\
\multicolumn{1}{l}{\sc{matplotlib}} &282 &1484 && ${\bf2.528}$ & 2.545 & 2.546 & \underline{2.589} \\
\multicolumn{7}{l}{{\em Active Directory graphs - Directed}}\\
\multicolumn{1}{l}{\sc{r2000}} &5997 &18795 && ${\bf1.924}$ & ${\bf1.924}$ & ${\bf1.924}$ & ${\bf1.924}$ \\
\multicolumn{1}{l}{\sc{r4000}} &12001 &45781 && ${\bf3.967}$ & 4.057 & 4.003 & 3.970 \\
\multicolumn{1}{l}{\sc{ads5}} &1523 &5359 && ${\bf3.499}$ & 3.510 & 3.502 & 3.510 \\
\multicolumn{1}{l}{\sc{ads10}} &3015 &12776 && ${\bf4.629}$ & 4.876 & 4.678 & 4.658 \\
\bottomrule

\end{tabular}

\caption{Expected query count under different algorithms: {\sc Other}=Best out of 8 heuristics from literature (those underlined are {\bf not} achieved by {\sc H1}); {\sc Tree}=Heuristic based on the exact algorithm; {\sc RL} and {\sc MCTS} are self-explanatory}

\label{tablebig20}

\end{table*}

\subsection*{Graph description}

In this section, we include the sources of the graphs used in our experiments.

\begin{itemize}

    \item {\sc USAir97}~\cite{Davis11:University}: An edge
        represents that there are direct flights between two airports.

    \item {\sc euroroad}~\cite{konect} and {\sc minnesota}~\cite{Davis11:University}:
        A mostly Europe road network and Minnesota's road network.

    \item {\sc Power}~\cite{Rossi15:Network,Watts98:Collective}: Western States
        Power Grid of the US.

    \item {\sc bcspwr} series~\cite{bcspwr}:
        Power network patterns collected by Boeing Computer Services.

    \item {\sc bugfix}~\cite{Allard19:Dataset}: A directed graph describing
        the workflow of bugfix.

    \item {\sc football}~\cite{Davis11:University}: Nodes are countries. Directed
        edges represent export of football players between countries.

    \item {\sc chesapeake}~\cite{konect}: Nodes are organism types in Chesapeake, US.
        Edges denote carbon exchange.

    \item {\sc cattle}~\cite{konect}: Directed edges represent dominance behaviour observed between cattle in Iberia Livestock Experiment Station in Louisiana, US.

    \item Python package dependency graphs~\cite{pydeps}:
        We generate Python package dependency graphs for popular Python packages using an open source tool {\sc pydeps}.

    \item Active Directory attack
        graphs~\cite{dbcreator,adsimulator,Guo22:Practical,Goel22:Defending,Guo23:Scalable}: These
        are graphs generated using two different open source synthetic Active
        Directory graph generators ({\sc DBCreator}~\cite{dbcreator} and {\sc adsimulator}~\cite{adsimulator}).
        The selection of parameters and preprocessing steps follow
        \cite{Guo22:Practical,Goel22:Defending,Guo23:Scalable}.

\end{itemize}

We conclude with two histograms Figure~\ref{fig:euroroadhist} and Figure~\ref{fig:r4000hist}. We set $p=0.5$ and use the same sources and destinations from the other experiments. Please refer to the reproducibility notes on how sources and destinations are selected. We run {\sc H1} on {\sc euroroad} (undirected) and {\sc r4000} (directed) {\em without query limit} for $1000$ episodes each. We show the query count distribution of these two graphs.
Both exhibit clear {\em heavy-tail} patterns.
For
{\sc euroroad}, most episodes end in less than $10$ queries and the vast majority
actually end within $5$ queries.
For {\sc r4000}, even though some episodes can go as long as $50$ queries, still, most of the episodes end within $20$ queries and the vast majority end within
$10$ queries.

\begin{figure}[h]
    \includegraphics[width=0.5\linewidth]{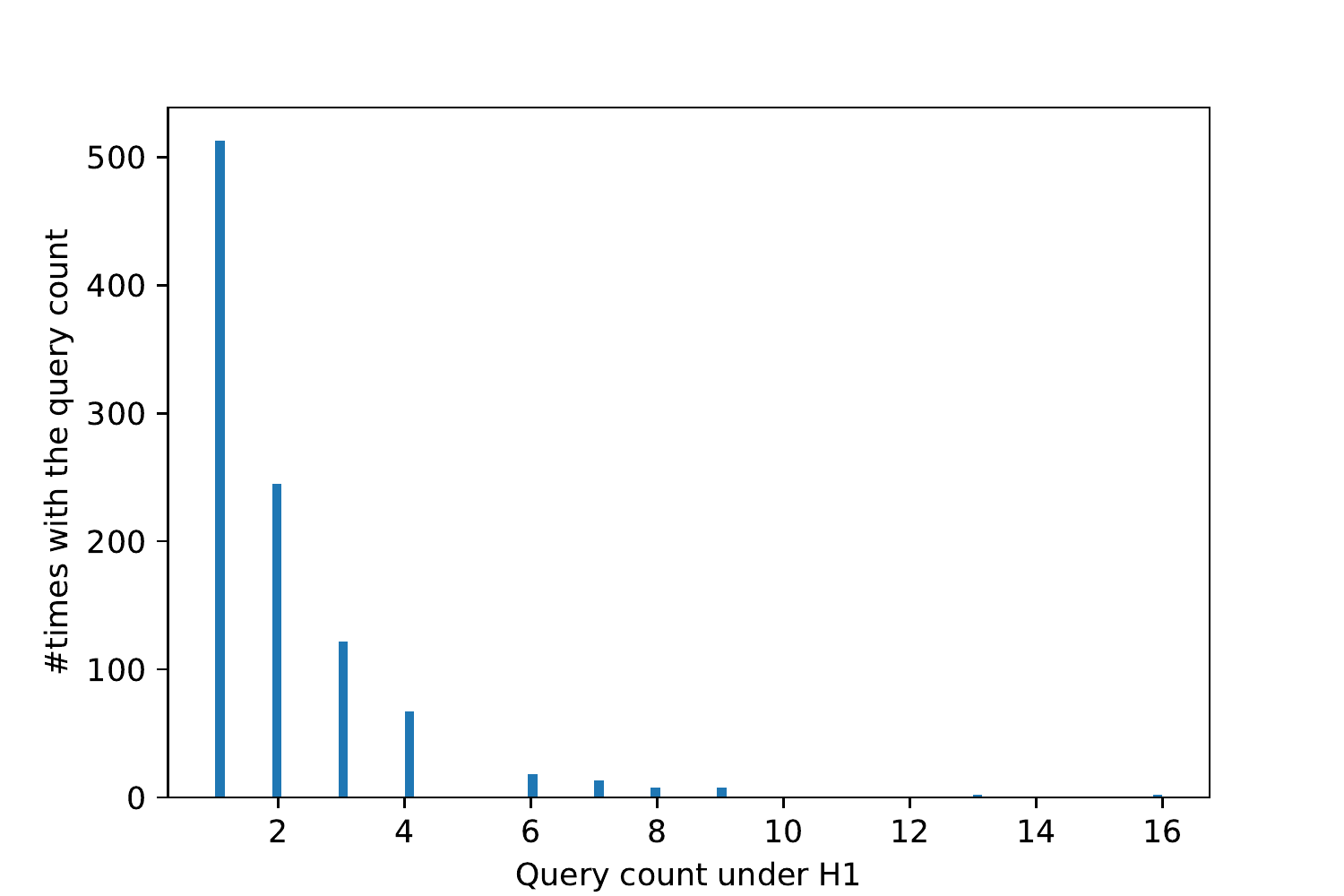}
\centering
    \caption{{\sc H1} on {\sc euroroad} without query limit ($1000$ episodes)}
    \label{fig:euroroadhist}
\end{figure}

\begin{figure}[h]
    \includegraphics[width=0.5\linewidth]{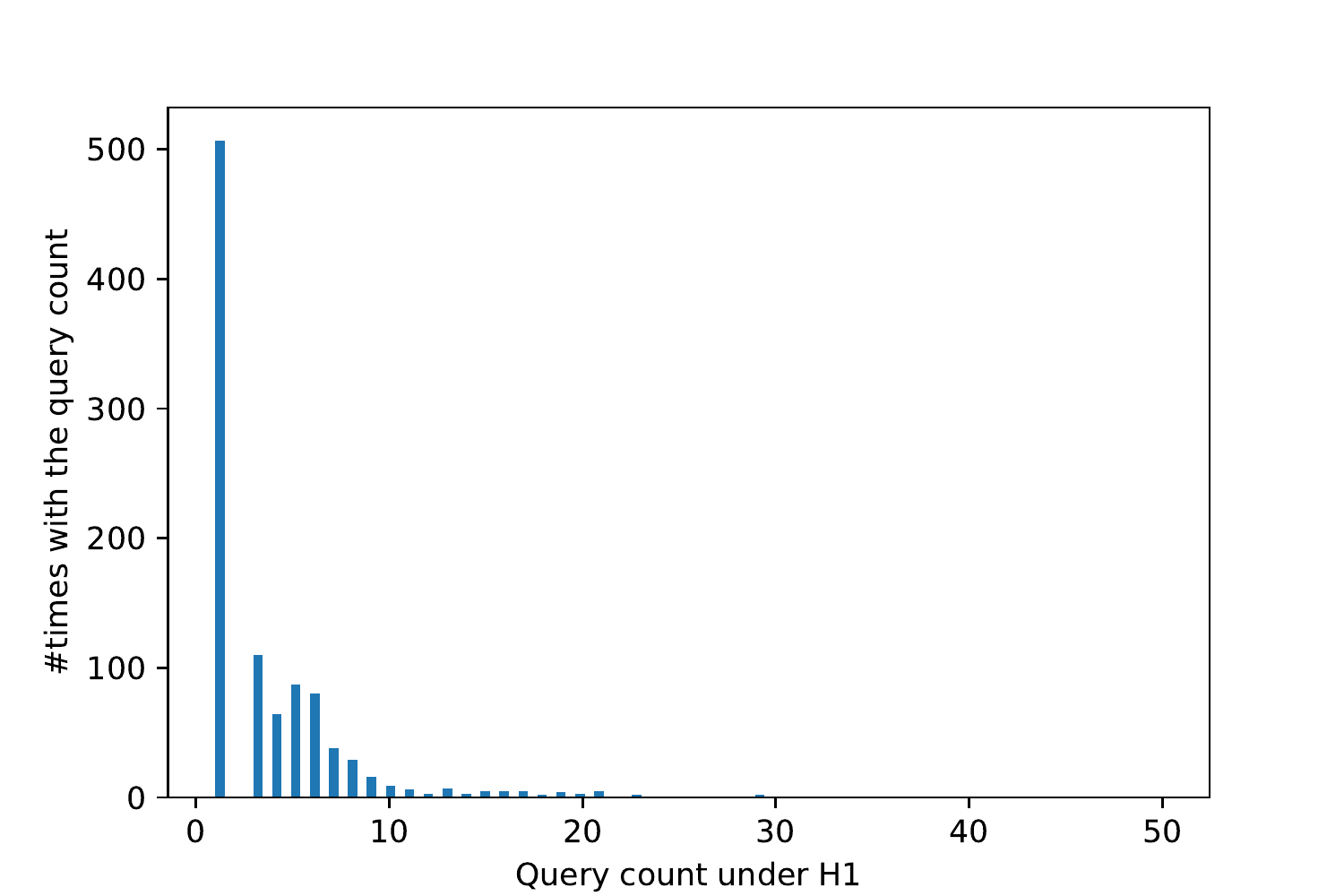}
\centering
    \caption{{\sc H1} on {\sc r4000} without query limit ($1000$ episodes)}
    \label{fig:r4000hist}
\end{figure}

\subsection*{Heuristics from literature}

There are a long list of relevant heuristics and approximation
algorithms on sequential testing, learning with attribute costs and stochastic Boolean
function evaluation.
\cite{Unluyurt04:Sequential}'s survey mentioned
heuristics from~\cite{Jedrzejowicz83:Minimizing,Cox89:Heuristic,Sun96:Hill}.
Approximations algorithms were proposed in~\cite{Allen17:Evaluation,Deshpande14:Approximation,Kaplan05:Learning,Golovin11:Adaptive}.
We {\em select} and {\em port} $8$ heuristics from literature and used them in our experimental comparison.

\vspace{.1in}
\noindent{\em On selection:} Some existing heuristics do not apply to our model (i.e., the greedy heuristic
that queries the cheapest bit --- our model assumes unit cost). Some heuristics
are too expensive for large graphs.  For example, one heuristic from
\cite{Jedrzejowicz83:Minimizing} requires computing the \#P-hard graph
reliability problem for {\em every} edge for {\em every} query, which
is not feasible when there are tens of thousands of edges.

\vspace{.1in}
\noindent{\em On porting:}
When we implement a heuristic/an approximation algorithm from existing
literature, we took the liberty and made necessary modifications. For example, the
approximation algorithm in~\cite{Deshpande14:Approximation} involves counting
the number of paths/cuts left for $m$ times to make one query decision,
which is computationally infeasible. Instead, we heuristically
generate a sample set of paths/cuts and only count how many from
our sample sets are left.
As another example, the approximation algorithm in~\cite{Allen17:Evaluation}
would query paths/cuts in one go (continue querying even after a path/cut
is eliminated). This is fine for an approximation algorithm, but to use it as a heuristic,
we switch to not continuing querying eliminated paths/cuts.
\vspace{.1in}

Besides {\sc
H1} (Definition~\ref{def:h1}~\cite{Jedrzejowicz83:Minimizing,Cox89:Heuristic}),
we also implemented the following $7$ heuristics.

\begin{itemize}

    \item {\sc H2-Both}, {\sc H2-Path} and {\sc H2-Cut}, inspired by \cite{Kaplan05:Learning}: \cite{Kaplan05:Learning}
        proposed an approximation algorithm for learning with attribute costs.  The idea is to alternate between
        two greedy algorithms.  In the context of our model, one greedy
        algorithm greedily queries the edge that appears in the most paths.
        The second greedy algorithm greedily queries the edge that appears in
        the most cuts.  Basically, one greedy algorithm is banking on that a
        path exists.  Similarly, the other one expects that a cut exits.  One
        of the two algorithms is doing wasteful work but that is fine for
        an approximation algorithm (as it only increases the cost by
        a constant factor of $2$).

        We implemented the following three heuristics:

        {\sc H2-Both} greedily queries the edge
        that appears in most certificates (either paths or cuts).

        {\sc H2-Path} greedily queries the edge
        that appears in most paths.

        {\sc H2-Cut} greedily queries the edge
        that appears in most cuts.

        Lastly, instead of enumerating all paths/cuts, we generate
        a sample set of paths and cuts. Greedy selection is conducted
        based on the samples.

        Path and cut samples are generated as follows:

        \begin{itemize}

            \item We generate {\sc H1}'s query policy tree by setting a query limit of $10$. Since {\sc H1}'s decision is always the intersection of a path
                and  a cut, we collect all paths and cuts ever referenced
                by {\sc H1}'s policy tree. Presumably, these paths and cuts
                are more relevant.

            \item We also add the following paths and cuts into our samples.
                We start with the shortest path $p$.
                We disallow one edge on it, and obtain the alternative shortest path $p'$ . We continue this way. We disallow one edge on $p'$,
                and obtain another shortest path $p''$.
                We exhaustively generate all paths that can be obtained
                this way, starting from only disallowing one edge, to disallowing
                two edges, etc. until we collect $100$ paths.
                We do the same for cuts.

        \end{itemize}

        For all remaining heuristics, whenever they reference all paths or all
        cuts, in our implementation, we use the above samples instead.

    \item {\sc MinSC}, {\sc MinSC-Path}, {\sc MinSC-Cut}, inspired by \cite{Allen17:Evaluation}: \cite{Allen17:Evaluation} proposed an approximation algorithm for
        $k$-term DNF style SBFE. The overall algorithm also alternates between two algorithms.
        In the context of our model, one algorithm tries to approximately
        solve a minimum set cover for covering all paths (or covering all cuts) and tries to find the edges that appear in the minimum set cover.
        The other algorithm greedily selects the shortest path or the minimum cut
        and always query the path/cut in one go (i.e., query every edge along
        the shortest path/minimum cut and does not stop even though the
        path/cut is already eliminated), which is fine for an
        approximation algorithm.
        In our implementation, we made two changes. We stop querying a path/cut
        once it is eliminated.
        Also, we solve the minimum set cover problem optimally instead of
        approximately, as modern MIP solvers
        like Gurobi can solve small set cover instances sufficiently fast.

        We implemented the following three heuristics:

        We optimally solve the minimum set cover problem for covering all paths
        and for covering all cuts.  Edges in $E_P$ cover all paths and edges in
        $E_C$ cover all cuts.

        {\sc MinSC} queries an edge that is in the intersection of $E_P$ and $E_C$.

        {\sc MinSC-Path} queries the edge that is in the intersection
        of $E_P$ and the minimum cut.

        {\sc MinSC-Cut} queries the edge that is in the intersection
        of $E_C$ and the shortest path.

    \item {\sc Adaptive Submodular}~\cite{Deshpande14:Approximation,Golovin11:Adaptive,Fu17:Determining}: We query the edge that minimizes the product of ``number
        of paths remain if the query response is {\sc Off}'' and ``number of cuts remain if the query response is {\sc On}''.

\end{itemize}

Among $75$ experiments ($25$ graphs with $B\in\{5,10,20\}$),
{\sc H1} only lost to the other heuristics $7$ times.
\
\begin{itemize}

    \item {\sc MinSC-Path} performs the best for {\sc minnesota} and $B=10$.

    \item {\sc MinSC Cut} performs the best for {\sc bugfix} and $B=20$.

    \item
        {\sc MinSC} performs the best for two graphs when $B=10$ ({\sc cattle} and {\sc matplotlib}). It also performs the best for four graphs when $B=20$ ({\sc bcspwr01}, {\sc bcspwr03} and {\sc matplotlib}).

\end{itemize}

\subsection*{Reproducibility notes}

\noindent{\em On picking source and destination:} We use seed $0$ to $10$ ($11$ seeds in total). We randomly draw a source node $s$. Among all nodes reachable from the source,
        we randomly pick a destination node $t$.
        If $s$ cannot reach any other node, then we redraw.
        For Active Directory graphs, we always set the destination to Domain Admin (a particular node
        representing the highest privilege).
        The source node is then randomly drawn from those that have available attack paths to Domain Admin.

        Among the $11$ seeds, some result in trivial instances with small expected query count (i.e., the extreme cases only require $1$ query).
        Generally speaking, higher expected query count corresponds to higher computational demand (i.e., more paths and cuts that need to be considered
        and larger tree structure).
        To measure the ``median difficulty'' of a graph, we pick the
        ``median seed'' from $11$ seeds.
        Specifically, we run {\sc H1} with a query limit of $B=10$ on all $11$ seeds, and record $11$ expected query counts. ``Median seed'' is the seed that corresponds to the median of these $11$ data points (i.e., the seed that corresponds
        to ``median difficulty'' for {\sc H1} --- considering that
        all heuristics' performances are close, the chosen seed
        should lead to a somewhat representative level of ``median difficulty'' for
        other heuristics).
        Besides for avoiding extreme cases including trivial cases, we pick only one seed for each graph mainly
        to save computational time. We allocate $72$ hours for the exact algorithm
        and we have $25$ graphs. Our current experiments (using one seed for each graph) already make use of a high-performance cluster.

    \vspace{.1in}
 \noindent{\em On picking random samples for comparing heuristics:}
        When there are $m$ edges, the total number of possible graph states
        is $2^m$ (as each edge is either {\sc On} or {\sc Off}). We cannot afford to directly sample the graph states as it is too large.

        When the query limit $B$ is small, we actually can evaluate the heuristics
        {\em exactly} without resorting to random sampling at all.
        When we apply a query policy, the query cost only depends on the
        query answers received. If $99\%$ of the edges are not even queried, then their states are irrelevant, and we shouldn't spend effort sampling them. The query results can be described as binary vectors of size $B-1$.
        The $i$-th bit describes the $i$-th query's result.
        The $B$-th bit is not needed as the last query's result has no impact
        on the total query cost (we stop regardless anyway).
        For example, if $B=4$, then the query results are just one of the following $8$ types:
        $<0,0,0>$,
        $<0,0,1>$,
        $<0,1,0>$,
        $<0,1,1>$,
        $<1,0,0>$,
        $<1,0,1>$,
        $<1,1,0>$,
        $<1,1,1>$.
        If $B$ is small, then we can simply exhaustively go over all possible
        query results as there are only $2^{B-1}$ possibilities.
        In our experiments, for $B=5$ and $B=10$, we evaluate the heuristics
        exactly using this method.

        When $B$ is large, we do need to rely on random samples.
        To improve efficiency on sample evaluation, for example, for more efficient evaluation of {\sc Tree}, we draw samples in the following way.
        Suppose we have drawn a vector of size $B-1$:

        \[<\underbrace{0,1,0,1,1,\ldots,0}_{B-5},\underbrace{1,1,0,1}_4>\]

        Suppose we always
        pretend that there are only $5$ queries left.  Due to the nature of our
        proposed heuristic, when we get close to the query limit, i.e., when
        there actually $5$ queries left, the policy tree $\pi$ we generated can
        handle all scenarios that can possibility arise from this point on.
        That is, for the response vector shown above, the last $4$ bits do not
        affect $\pi$ at all and $\pi$ can handle all possible trailing bits (of length $4$).
        So we can expand the above response vector into $16$ vectors (do not
        change the first $B-5$ bits and use all combinations for the last $4$
        bits):

        \[<\underbrace{0,1,0,1,1,\ldots,0}_{B-5},\underbrace{0,0,0,0}_4>\]
        \[<\underbrace{0,1,0,1,1,\ldots,0}_{B-5},\underbrace{0,0,0,1}_4>\]
        \[<\underbrace{0,1,0,1,1,\ldots,0}_{B-5},\underbrace{0,0,1,0}_4>\]
        \[<\underbrace{0,1,0,1,1,\ldots,0}_{B-5},\underbrace{0,0,1,1}_4>\]
        \[\vdots\]
        \[<\underbrace{0,1,0,1,1,\ldots,0}_{B-5},\underbrace{1,1,1,1}_4>\]

        Once $\pi$ is derived, the expected query cost on the above $16$ vectors
        is already computed (essentially it is our integer program's objective
        value). That is, we compute once for the first part of the response
        vector (first $B-5$ bits) and we get the expected performance of $16$
        data points, which results in increased efficiency for evaluating samples.

        For the above reason, we generate the response vectors by generating
        $1000$ vectors of size $B-5$ without replacement. For each, we pad it
        into $16$ complete response vectors.  All together, we have $16,000$
        response vectors.  We always use the same seed $0$ for generating the
        vectors.  And once the vector set is generated, all heuristics use the
        same set of samples.

In our experiments, our heuristic won $11$ times and lost $2$ times
against the best other heuristic when $B=10$ (evaluated exactly without resorting
to random sampling).
Our heuristic won $13$ times and lost $6$ times against the best other heuristic
when $B=20$ (evaluated using random sampling).
We lost a few more times when $B=20$, which perhaps can be explained by
the additional randomness. 

\subsection{Hyper-parameters behind our heuristics}

We list all hyper-parameters used in our experiments.

{\sc Tree}: We record the best result for $B'\in\{4,5,6\}$.

{\sc RL}: We use $B'\in \{2,3\}$ to generate the action space, corresponding to an action space size of $3$ and $7$.
We applied the discrete SAC implementation from {\sc Tianshou}.  The hyper-parameters are as follows: actor learning rate is $0.0001$;
critic learning rate is $0.001$; number of epochs is $1000$; step per epoch is $1000$; step per collect is $20$; batch size is $128$; network
size is $[64,64]$.

{\sc MCTS}: The action space is the same as {\sc RL}. We apply epsilon-greedy MCTS with $\epsilon=0.2$. We run $1000$ simulations to generate
the next query.
When $B=20$, we have to reduce the simulation count for some graphs, as we only allow $72$ hours to run
MCTS for $16,000$ episodes.
For {\sc minnesota} and {\sc power}, we used $500$ simulations to generate the next query.
For {\sc bcspwr09} and {\sc bcspwr10}, we used $200$ and $50$ simulations to generate the next query, respectively.

\begin{figure}[t]
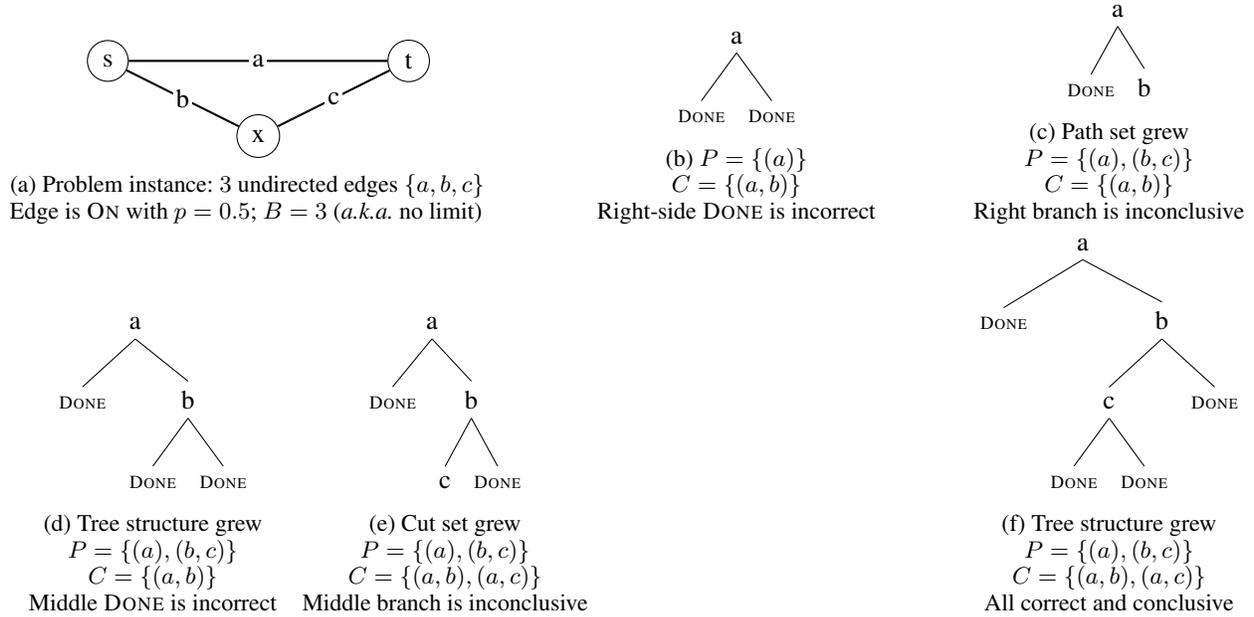

     \centering
     \begin{subfigure}[b]{0.40\textwidth}
         \centering
\tikz \graph [edge quotes={fill=white,inner sep=1pt},
          grow down, branch right, nodes={circle,draw}] {
    s --[thick, "a"] t[at={(4,1)}];
    s --[thick, "b"] x[at={(1,-1)}] --[thick, "c"] t
         };
         \caption{Problem instance: $3$ undirected edges $\{a,b,c\}$\\
         Edge is {\sc On} with $p=0.5$; $B=3$ ({\em a.k.a.} no limit)}
         \label{fig:example1leftclone}
     \end{subfigure}
     \hfill
     \begin{subfigure}[b]{0.23\textwidth}
\captionsetup{justification=centering}
         \centering
         \Tree [.a [.{\scriptsize {\sc Done}} ] [.{\scriptsize {\sc Done}} ] ]
         \caption{$P=\{(a)\}$\\$C=\{(a,b)\}$\\Right-side {\sc Done} is incorrect}
         \label{fig:step1}
     \end{subfigure}
     \hfill
     \begin{subfigure}[b]{0.23\textwidth}
\captionsetup{justification=centering}
         \centering
         \Tree [.a [.{\scriptsize {\sc Done}} ] [.b ] ]
         \caption{Path set grew\\
         $P=\{(a),(b,c)\}$\\$C=\{(a,b)\}$\\Right branch is inconclusive}
         \label{fig:step2}
     \end{subfigure}
     \hfill
     \begin{subfigure}[b]{0.23\textwidth}
\captionsetup{justification=centering}
         \centering
         \Tree [.a [.{\scriptsize {\sc Done}} ]
         [.b [.{\scriptsize {\sc Done}} ] [.{\scriptsize {\sc Done}} ] ]
]
         \caption{Tree structure grew\\
         $P=\{(a),(b,c)\}$\\$C=\{(a,b)\}$\\Middle {\sc Done} is incorrect}
         \label{fig:step3}
     \end{subfigure}
     \begin{subfigure}[b]{0.23\textwidth}
\captionsetup{justification=centering}
         \centering
         \Tree [.a [.{\scriptsize {\sc Done}} ]
         [.b [.c ] [.{\scriptsize {\sc Done}} ] ]
]
         \caption{Cut set grew\\
         $P=\{(a),(b,c)\}$\\$C=\{(a,b),(a,c)\}$\\Middle branch is inconclusive}
         \label{fig:step4}
     \end{subfigure}
     \hfill
     \begin{subfigure}[b]{0.23\textwidth}
\captionsetup{justification=centering}
         \centering
         \Tree [.a [.{\scriptsize {\sc Done}} ]
         [.b [.c [.{\scriptsize {\sc Done}} ] [.{\scriptsize {\sc Done}} ] ] [.{\scriptsize {\sc Done}} ] ]
]
         \caption{Tree structure grew\\
         $P=\{(a),(b,c)\}$\\$C=\{(a,b),(a,c)\}$\\All correct and conclusive}
         \label{fig:step5}
     \end{subfigure}
        \caption{A walk through of our exact algorithm}
        \label{fig:step1to5}
\end{figure}

\subsection*{A walk through of our exact algorithm}

In Figure~\ref{fig:step1to5}, we walk through our exact algorithm.
The problem instance is described in Figure~\ref{fig:example1leftclone},
        which is the same instance as the one in Figure~\ref{fig:example1left}.

As mentioned toward the end of our description of the exact algorithm, the ``outer'' and ``inner'' loops in our algorithm description are purely for presentation purposes. In implementation, in every iteration, we generate the optimal correct policy tree given the current path set $P$, the current cut set $C$ and the current tree structure $S$.
If the resulting policy tree makes any wrong claims (i.e., claiming {\sc Done} prematurely), then we
expand the path/cut set. If the tree goes into any inconclusive situation (i.e., reaching
the leave position and still cannot claim {\sc Done}), then
we expand the tree structure. The final converged policy must be optimal.

\begin{itemize}


    \item We initialize with a path set containing one shortest path $P=\{(a)\}$
        and a cut set containing one minimum cut $C=\{(a,b)\}$.
        We initialize the tree structure to 2 layers complete.

    \item Iteration 1 (Figure~\ref{fig:step1}): The optimal policy tree after the first iteration is illustrated. The right-side {\sc Done} is incorrect. Confirming
        $a$ is {\sc Off} does not disprove all paths. We add the missing path $(b,c)$ into $P$.

    \item Iteration 2 (Figure~\ref{fig:step2}): The right branch is inconclusive.
        We expand the right branch.

    \item Iteration 3 (Figure~\ref{fig:step3}): The middle {\sc Done} is incorrect.
        Confirming $b$ is {\sc On} does not disprove all cuts. We add
        the missing cut $(a,c)$ into $C$.

    \item Iteration 4 (Figure~\ref{fig:step4}): The middle branch is inconclusive.
        We expand the middle branch.

    \item Iteration 5 (Figure~\ref{fig:step5}): All leaves are {\sc Done} and are all correct. This is the optimal policy tree identical to Figure~\ref{fig:example1right}.

\end{itemize}

\begin{figure}[t]
    \includegraphics[width=0.6\linewidth]{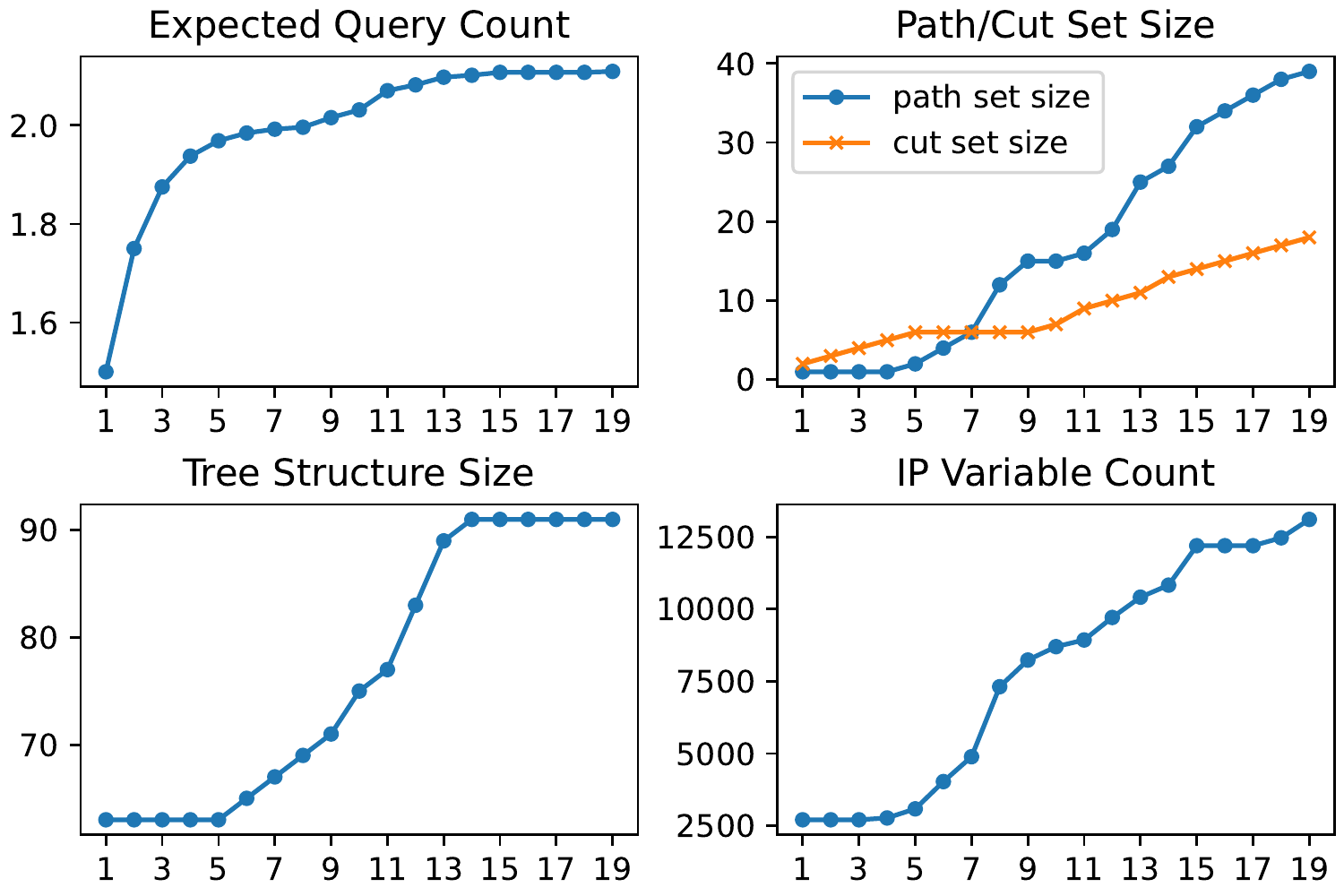}
\centering
    \caption{Exact algorithm on {\sc euroroad} with $B=10$: successfully find the optimal solution in 19 iterations}
    \label{fig:euroroad}
\end{figure}

\begin{figure}[t]
    \includegraphics[width=0.6\linewidth]{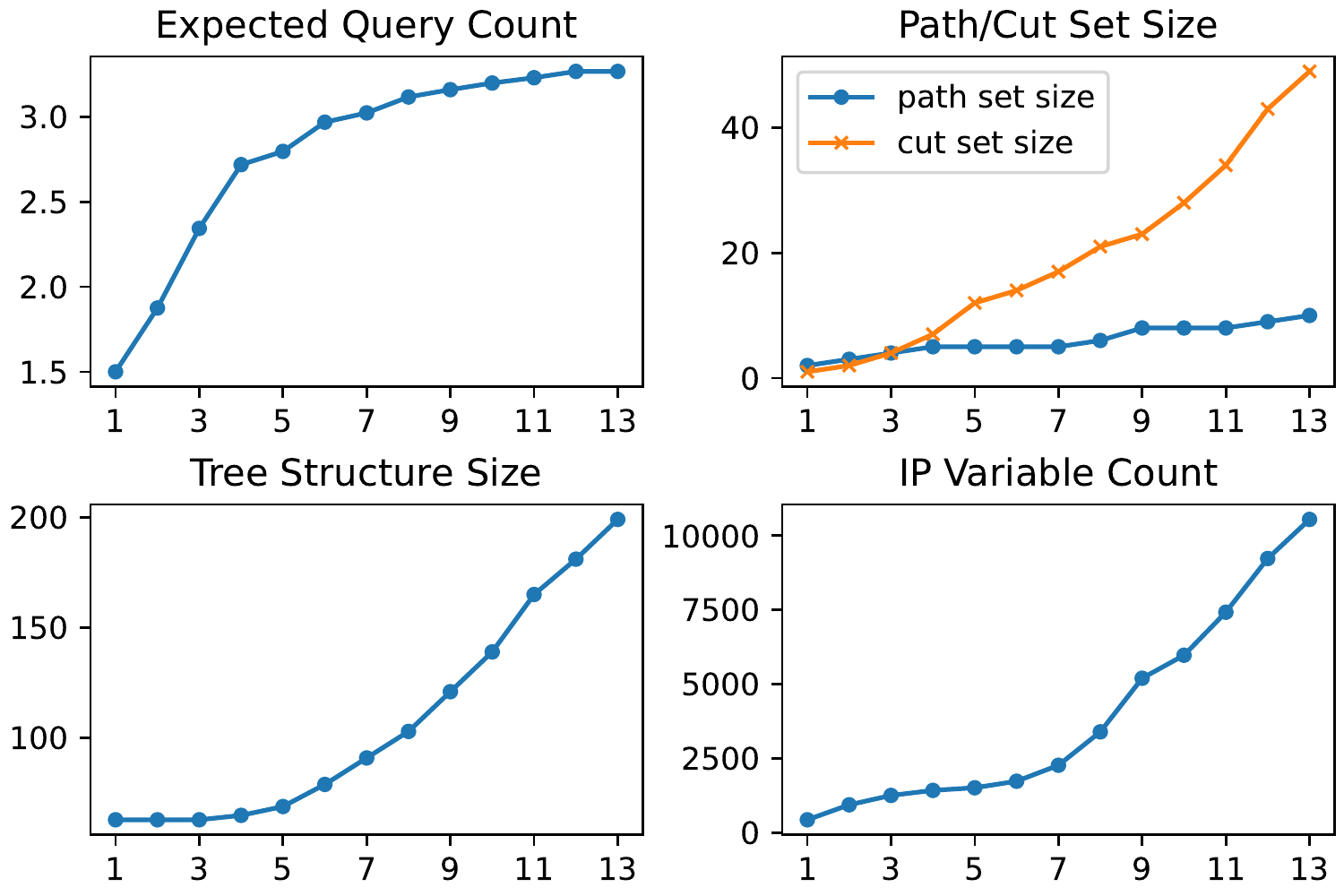}
\centering
    \caption{Exact algorithm on {\sc r4000} with $B=10$: time out (72 hours) after 13 iterations, find a lower bound that is at most $2.3\%$ away from optimality}
    \label{fig:r4000}
\end{figure}

\begin{figure}[t]
    \includegraphics[width=0.4\linewidth]{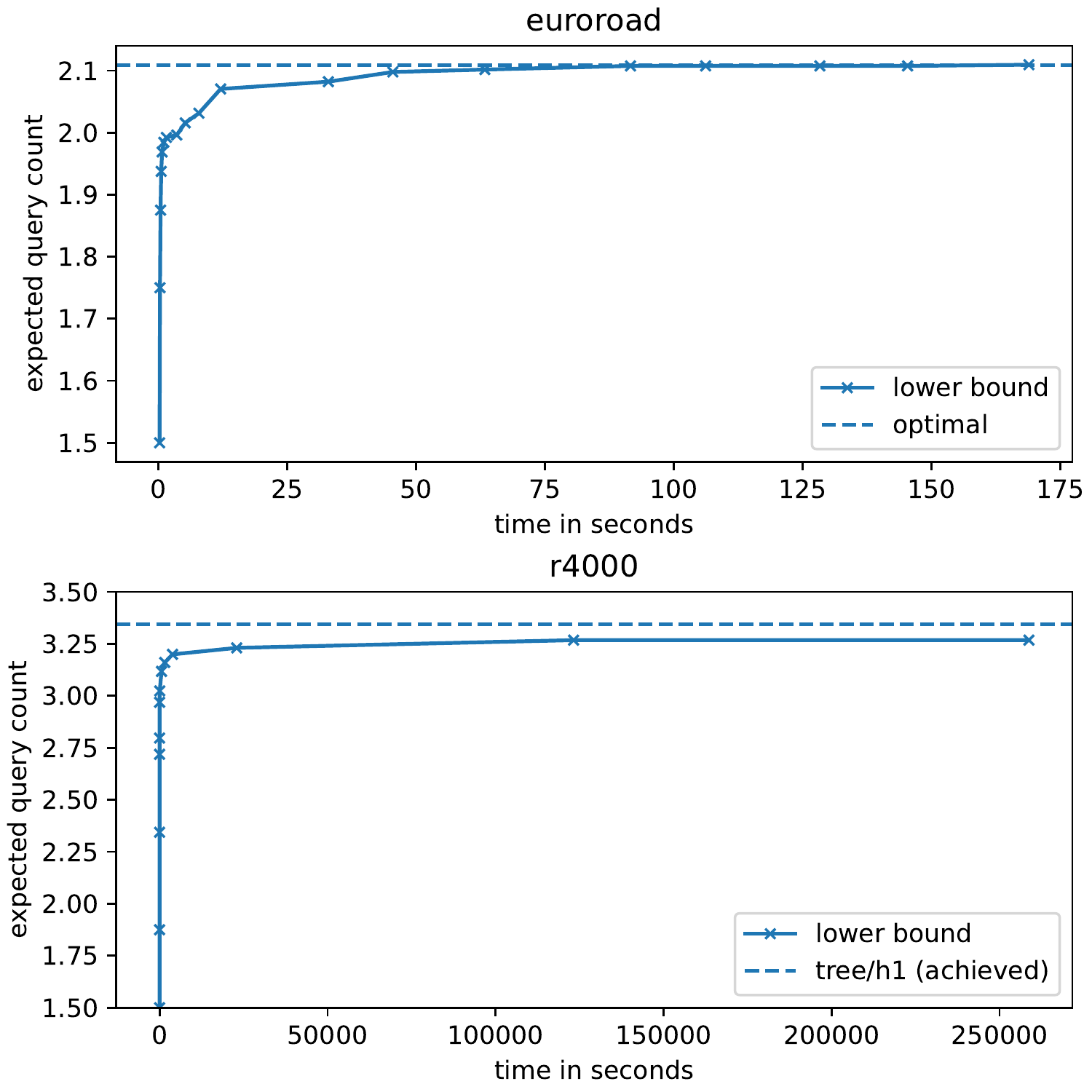}
\centering
    \caption{Exact algorithm on {\sc euroroad} and {\sc r4000} with $B=10$: {\sc euroroad} converged in $3$ minutes and the final result is proven optimal; {\sc r4000}
    took $72$ hours to get to $2.3\%$ away from results achieved by heuristics.}
    \label{fig:time}
\end{figure}

\subsection*{Running details on two large graphs}

In Figure~\ref{fig:euroroad} and Figure~\ref{fig:r4000}, we present the
running details of our exact algorithm on two graphs with a query limit of $10$.

Figure~\ref{fig:euroroad} is on {\sc euroroad} (undirected, $1175$ nodes and
$1417$ edges). The exact algorithm terminates in $19$ iterations, successfully
identifying the optimal query policy.  The final path set has a size of
$|P|=39$.  The final cut set has a size of $|C|=18$.  These numbers show that
our exact algorithm successfully managed to pick out only a small number of
relevant paths and cuts for decision making.  The final tree structure size is
$91$. Considering that when the algorithm finishes, the leaf nodes must be
{\sc Done} (except when reaching the query limit), the number of query nodes in the
final policy tree should be much less than $91$.
On the
contrary, if we do not use our ``tree structure growth'' idea, and simply use a complete
tree, then the number of query nodes is $2^{10}-1=1023$.
Our tree growth idea managed to significantly
reduce the search space (from making about $1023$ decisions to much less than $91$ decisions).

Figure~\ref{fig:r4000} is on {\sc r4000} (directed, $12001$ nodes and $45781$
edges). The exact algorithm failed to terminate but it successfully produced a
close lower bound $3.268$ ($3.344$ is achieved by our heuristic, so
the lower bound is at most $2.3\%$ away from optimality).
The main scalability bottleneck for this graph is the tree size. It gets to $199$
at the end of $72$ hours.
Because the algorithm did not terminate, not all leaf nodes were {\sc Done}
among those $199$ nodes.
The number of query nodes is a lot more than that of {\sc euroroad}.

Finally, in Figure~\ref{fig:time}, for {\sc euroroad} and {\sc r4000} with $B=10$,
we plot the lower bound (expected query count after each iteration of the exact algorithm) against time spent by
the exact algorithm.
For {\sc euroroad}, it took less than $3$ minutes to reach confirmed optimality.
For {\sc r4000}, it took $72$ hours (time out) to reach within $2.3\%$ away
from results achieved by heuristics,
which is also within $2.3\%$ away from optimality.

\begin{algorithm}[h!]
\caption*{Exact Algorithm for Limited Query Graph Connectivity Test}
\textbf{Input}: Graph $G$ with source $s$ and destination $t$\\
    Query limit $B$\\
    Initial path set $P$ (i.e., a singleton set with one shortest path)\\
    Initial cut set $C$ (i.e., a singleton set with one minimum cut)\\
    Initial tree structure $S$ (i.e., root node only)\\
\begin{algorithmic}[1]
    \WHILE {{\sc True}}
    \STATE {\sc Converged} $\leftarrow$ {\sc True}
    \STATE Generate the optimal correct policy tree $T(S,P,C)$ via integer program
    with input $S$, $P$ and $C$.
    \FOR {every node $i$ in $T(S,P,C)$}
        \IF {node $i$ claimed {\sc Done} by disproving all paths in $P$ \AND node $i$ has not established a cut}
        \STATE Find a path that is still possible (involving only
        {\sc On} or hidden edges) via Dijkstra's shortest path algorithm (assign high weights on {\sc Off} edges). Add this path to $P$.
        \STATE {\sc Converged} $\leftarrow$ {\sc False}
        \ENDIF
        \IF {node $i$ claimed {\sc Done} by disproving all cuts in $C$ \AND node $i$ has not established a path}
        \STATE Find a cut that is still possible (involving only
        {\sc Off} or hidden edges) via weighted minimum cut (assign high weights on {\sc On} edges). Add this cut to $C$.
        \STATE {\sc Converged} $\leftarrow$ {\sc False}
        \ENDIF
        \IF {leaf node $i$ is inconclusive (i.e., not reaching the query limit and
        not {\sc Done}}
        \STATE Expand $S$ by attaching two child nodes to $i$.
        \STATE {\sc Converged} $\leftarrow$ {\sc False}
        \ENDIF
    \ENDFOR
        \IF {{\sc Converged}}
        \RETURN $T(S,P,C)$
        \ELSE
        \STATE Enter the next iteration with updated $S$, $P$, $C$.
        \ENDIF
\ENDWHILE
\end{algorithmic}
\end{algorithm}

\subsection*{Pseudocode and integer program details}

We first describe the integer program, whose task is to optimally place
the queries into a given tree structure, while ensuring that any {\sc Done}
node placed either disproves the given path set $P$ or disproves the given
cut set $C$ (i.e., the output is a correct policy tree). The goal is to minimize the resulting tree's cost.

\noindent {\bf Integer program's input:} Path set $P$, Cut set $C$, Tree structure $S$

\noindent {\bf Variables:} Let $E^R$ be the set of edges referenced in either $P$ or $C$.

\begin{itemize}

    \item
For every tree node $i\in S$, for every edge $e\in E^R$,
binary variable $v_{e,i}$ being $1$ means we will perform query $e$ at node $i$.

    \item
        For every tree node $i\in S$, binary variable $v_{\textsc{Done},i}$ being $1$ means it is correct to claim {\sc Done} at node $i$.

\end{itemize}

\noindent {\bf Minimization objective:} Let $\textsc{Prob}(i)$ be the probability
of reaching node $i$. This is a constant depending only on the tree structure (therefore part of the input). For example, if it takes $3$ left turns ($3$ {\sc On} results) and $2$ right turns (2 {\sc Off} results) to reach $i$, then
$\textsc{Prob}(i)=p^3(1-p)^2$.

\[\textsc{Minimize:}\quad \sum_{i\in T}\textsc{Prob}(i)\{\sum_{e\in E^R}v_{e,i}\}\]

\noindent {\bf Useful notation:}

\begin{itemize}

    \item $\textsc{Parent}(i)$ is the parent of node $i$.

    \item $\textsc{route}(i)$ is the set of nodes along the route from root to node $i$ (inclusive).

    \item $\textsc{lnodes}$ is the set of tree nodes who are left children of their parents.
    \item $\textsc{rnodes}$ is the set of tree nodes who are right children of their parents.

    \item For $i\in \textsc{lnodes}$,
\[N(i)=\{\textsc{Parent}(j)|j \in \textsc{lnodes}\cap \textsc{route}(i)\}\]
    \item For $i\in \textsc{rnodes}$,
\[N(i)=\{\textsc{Parent}(j)|j \in \textsc{rnodes}\cap \textsc{route}(i)\}\]

\end{itemize}

\noindent {\bf Subject to:}

\[\forall i, \sum_{e\in E^R}v_{e,i}+v_{\textsc{Done},i} =1\]

\[\forall i, v_{\textsc{Done},i}\ge v_{\textsc{Done},\textsc{Parent}(i)} \]

\[\forall i\in \textsc{Leaves},\forall e\in E^R, \sum_{j\in \textsc{route}(i)} v_{e,j}\le 1 \]

\[
    \forall\textsc{cut}\in C,\forall i \in \textsc{lnodes},\]
\[\sum_{j\in N(i),e\in\textsc{cut}}v_{e,j}-(v_{\textsc{Done},i}-v_{\textsc{Done},\textsc{Parent}(i)}) \ge 0
\]
\[
    \forall\textsc{path}\in P,\forall i \in \textsc{rnodes},\]
\[\sum_{j\in N(i),e\in\textsc{path}}v_{e,j}-(v_{\textsc{Done},i}-v_{\textsc{Done},\textsc{Parent}(i)}) \ge 0
\]

\subsection*{Proof of Theorem~\ref{thm:sharpp} (\#P-hardness)}

\cite{Fu17:Determining} already derived a \#P-hardness proof for their model.  In the authors' setting, the query cost of
an edge can be an exponential value.  {\em The assumption of exponential costs is core to the authors'
hardness proof, which is very restrictive and not applicable for many practical applications.} The main idea behind the authors' proof was that if there are two parallel subgraphs connecting the source and
the destination, and both involve an exponentially expensive edge and the
rest of the subgraph, then if we are capable of correctly choosing between the two parallel routes,
then that means we can accurately assess the reliability of subgraphs, which
is \#P-hard. Our proof works for unit cost.
Our result is stronger as it shows hardness for a more restrictive
setting.


\begin{figure}[h]
\centering
\includegraphics[width=0.4\linewidth]{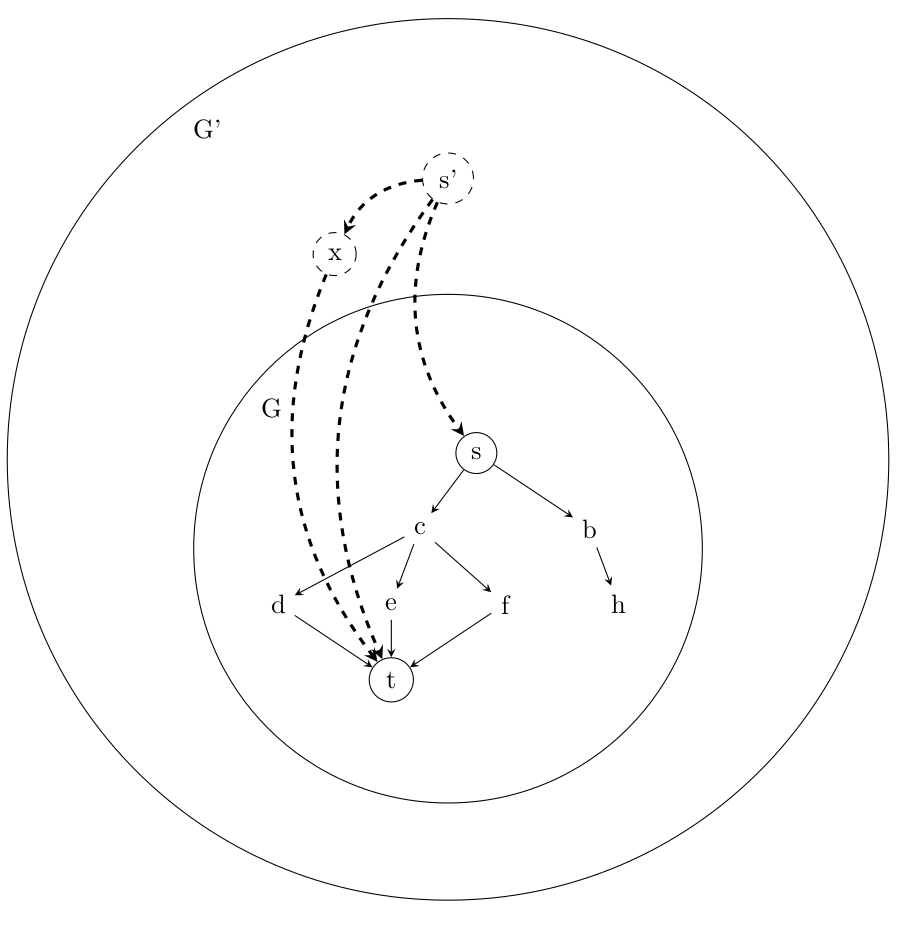}
    \caption{\#P-hardness proof illustration}
\label{fig:sharpp}
\end{figure}

\begin{proof}[Proof of Theorem~\ref{thm:sharpp}]

We reduce the task of computing the minimum expected number of queries to the network
    reliability problem~\cite{Valiant79:Complexity}, which is known to be \#P-hard.

    Let us consider an instance graph $G$ with $s$ and $t$ being the source
    and sink, respectively.
    We set $p=0.5$ and remove the query limit (i.e., set it to be the same as the number of edges).
    Let the number of edges of $G$ be $m$.
    As shown in Figure~\ref{fig:sharpp}, we construct
    a larger graph $G'$ containing $G$ by introducing two additional nodes $s'$ and
    $x$. We also introduce four additional edges as described below.
    It should be noted that $s'\rightarrow s$, $s'\rightarrow t$ and $s'\rightarrow x$ violate our uniform distribution assumption (i.e., their probabilities for
    being {\sc On} is not $0.5$). They are only temporary proof constructs for
    presentation purposes.
    We will remove them by the end of the proof using a simple idea:
    An edge that is {\sc On} with probability $1$ is simply replaced by
    $2m$ parallel edges\footnote{Adding parallel edges creates a multigraph.
    We could also remove the need for multigraphs by introducing auxiliary intermediate nodes.} that are {\sc On} with probability $0.5$.
    Similarly, an edge that is {\sc On} with probability $0$ is simply replaced by
    $2m$ serial edges that are {\sc On} with probability $0.5$.

    \begin{itemize}

        \item $s'\rightarrow s$, which is {\sc On} with probability $1$.
        \item $s'\rightarrow t$, which is {\sc On} with probability $0$.
        \item $s'\rightarrow x$, which is {\sc On} with probability $0$.
        \item $x\rightarrow t$, which is {\sc On} with probability $0.5$.

    \end{itemize}

    Let $q(G,s,t)$ be the minimum expected number of queries for instance
    $G$ with source $s$ and destination $t$.

    Let $q(G',s',t)$ be the minimum expected number of queries for instance
    $G'$ with source $s'$ and destination $t$.

    Let $Rel(G,s,t)$ be the reliability value of $G$ with
    source $s$ and destination $t$.  Defined in \cite{Valiant79:Complexity}, the
    reliability value $Rel(G,s,t)$ is the probability that there is a path from
    $s$ to $t$, assuming that every edge of $G$ has $50\%$ chance of
    missing. In the context of our model, $Rel(G,s,t)$ is simply the probability
    that there exists a path from $s$ to $t$.

    Now let us consider how to compute $q(G',s',t)$.
    Edge $s'\rightarrow s$ is always {\sc On}, so we would query it if and only
    if there is a path from $s$ to $t$ in $G$, and $s'\rightarrow s$ will only be queried after we find a path.
    Edge $s'\rightarrow x$ and $s'\rightarrow t$ are both always {\sc Off},
    so we would query these two if and only if there is a cut in $G$ between $s$ and $t$, and these two edges will only be queried after we find a cut.
    In summary, the optimal policy for instance $G'$ (from $s'$ to $t$) is that we should
    first solve instance $G$ (from $s$ to $t$). Then based on whether a cut or a path
    exists, we query either one additional edge ($s'\rightarrow s$) or
    two additional edges ($s'\rightarrow x$ and $s'\rightarrow t$).


    Next, we construct graph $G''$ based on $G'$, where we replace
            $s'\rightarrow s$ by $2m$ parallel edges
            and
        replace $s'\rightarrow t$ and $s'\rightarrow x$ by $2m$ serial edges.
    Let us consider $q(G'',s',t)$'s lower and upper bounds.

    To calculate $q(G'',s',t)$'s lower bound, we assume that we have
    access to an oracle that informs us whether $s'\rightarrow s$, $s'\rightarrow t$
    and $s'\rightarrow x$ are {\sc On} before we start the query process.
    With $1-O(\frac{1}{2^{2m}})$ probability,
    these edges are as expected (i.e., we expect $s'\rightarrow s$ to be {\sc On},
    and we expect both $s'\rightarrow t$ and $s'\rightarrow x$ to be {\sc Off}).
    When these edges are as expected, like before,
        we should query $s'\rightarrow s$ only after establishing a path, and
        we should query both $s'\rightarrow t$ and $s'\rightarrow x$ only after establishing a cut.
        Querying $2m$ parallel/serial edges takes $2-O(\frac{1}{2^{2m}})$ queries in expectation.
    With $O(\frac{1}{2^{2m}})$ probability, the above three
    edges are not as expected, in which case we simply assume the query count is $0$ for
    the purpose of calculating the lower bound.

    In summary, $q(G'',s',t)$ has a lower bound of
    \[(1-O(\frac{1}{2^{2m}}))\{q(G,s,t)+ (2-O(\frac{1}{2^{2m}})) \cdot Rel(G,s,t) +\]
    \[2 (2-O(\frac{1}{2^{2m}}))\cdot (1-Rel(G,s,t))\}\]
    \[\ge q(G,s,t)+ (2-O(\frac{1}{2^{2m}})) \cdot Rel(G,s,t) + \]
    \[2 (2-O(\frac{1}{2^{2m}}))\cdot (1-Rel(G,s,t)) - O(\frac{m}{2^{2m}})\]
    \[=q(G,s,t)+ 2 \cdot Rel(G,s,t) + 4 \cdot (1-Rel(G,s,t)) - O(\frac{m}{2^{2m}})\]
    \[=q(G,s,t)- 2 \cdot Rel(G,s,t) + 4 - O(\frac{m}{2^{2m}})\]

    To calculate $q(G'',s',t)$'s upper bound, we simply construct a
    feasible policy.  The policy first solves $q(G,s,t)$.  We query
    $s'\rightarrow s$ only after establishing a path, and we query both
    $s'\rightarrow t$ and $s'\rightarrow x$ only after establishing a cut.
    If $s'\rightarrow s$, $s'\rightarrow t$ or $s'\rightarrow x$ are not
    as expected ($O(\frac{1}{2^{2m}})$ chance), then we simply query every edge.
    The expected query count under this policy is then an upper bound on $q(G'',s',t)$, which is
    \[q(G,s,t)+ (2-O(\frac{1}{2^{2m}})) \cdot Rel(G,s,t) + \]
    \[2 (2-O(\frac{1}{2^{2m}}))\cdot (1-Rel(G,s,t)) + O(\frac{m}{2^{2m}})\]
    \[=q(G,s,t)- 2 \cdot Rel(G,s,t) + 4 + O(\frac{m}{2^{2m}})\]

    Combing the lower and upper bound, we have that
    \[Rel(G,s,t)= \frac{q(G,s,t) - q(G'',s,t)+4}{2} \pm O(\frac{m}{2^{2m}})\]

    An error of $O(\frac{m}{2^{2m}})$ does not matter when it comes to calculating $Rel(G,s,t)$, simply because $Rel(G,s,t)$ can only take values from the following set $\{0,\frac{1}{2^m},\frac{2}{2^m},\ldots,1\}$.
    Therefore, our Limited Query Graph Connectivity Test problem is as difficult
    as the network reliability problem.

\end{proof}

\subsection{More Discussion on Related Research}

\subsubsection{Stochastic Boolean Function Evaluation}

If we remove the query limit (i.e., set $B$ to the total number of edges), then
our problem becomes a special case of (monotone) {\em Stochastic Boolean
Function Evaluation (SBFE)}~\cite{Allen17:Evaluation,Deshpande14:Approximation}.
A SBFE instance involves a boolean function $f$
with multiple {\em binary} inputs and one {\em binary} output.  The input bits
are initially hidden and their values follow independent but not necessarily
identical Bernoulli distributions.
Each input bit has a query cost. The task is
to query the input bits in an adaptive order, until there is enough
information to determine the output of $f$. The goal is to minimize the
expected query cost.

We illustrate how to describe the problem instance in
Figure~\ref{fig:example1left} as a {\em unit cost uniform distribution}\footnote{{\em unit cost}: every query costs the same; {\em uniform distribution}: the input bits follow an i.i.d. Bernoulli distribution.} SBFE where
the boolean function is a {\em Disjunctive Normal Form (DNF)}.
To describe our problem instances as SBFE, we treat {\sc On} as $0$ and
treat {\sc Off} as $1$.
The DNF is then $f(a,b,c)=(a\wedge b)\vee (a\wedge c)$.
The expression basically says that $f$ is $1$ if and only if a cut exits,
since $a\wedge b$ and $a\wedge c$ are the only cuts in Figure~\ref{fig:example1left}.
Similarly, we could come up with a DNF (or CNF) in terms of paths, noting that
$\neg a$ and $\neg b\wedge\neg c$ are the only paths:

\begin{center}
DNF: $\neg f(a,b,c)=(\neg a)\vee (\neg b\wedge\neg c)$

CNF: $f(a,b,c)=(a)\wedge ( b\vee c)$
\end{center}

Our $f$ is always {\em monotone}. For example, if $f=1$, meaning that a cut exists,
then flipping one more edge from {\sc On}=$0$ to {\sc Off}=$1$ will never create a path.

Even though our techniques are developed mainly for our graph-based model and
our algorithm references graph properties,
{\em our techniques should be directly applicable to monotone unit cost uniform distribution SBFE, with or without the query limit.}
The reason is that the few graph properties we rely on also
hold for monotone SBFE in general. For example, one property we use is that any path and any cut must
intersect, but this is true for monotone SBFE in general, in the sense that any $0$-certificate (a certificate that contains only $0$s, i.e., a path involving only {\sc On} edges in the context of our model) and any $1$-certificate (a certificate that contains only $1$s, i.e., a cut involving
only {\sc Off} edges in the context of our model) must intersect.

\subsubsection{Existing Graph Connectivity Query Models}

The model in \cite{Fu17:Determining} is the closest related. The author's model
is essentially non-unit cost non-uniform distribution graph connectivity test, certainly without the query limit.  \cite{Fu17:Determining}
experimentally compared the exact algorithm (same as
\cite{Cox89:Heuristic}) against several existing heuristics.
Due to scalability, the exact algorithm was only evaluated on graphs with $20$
edges.
\cite{Fu16:Are} studied a model where there is an Erdos-Renyi graph with a source and a destination.
The goal is to spend the minimum number of queries to check whether the source and the destination is connected. In the context of our model, essentially the authors focused on {\em complete} graphs where every edge
has the same chance for being {\sc On}.
The authors proved an optimal algorithm, which is equivalent to {\sc H1}~\cite{Jedrzejowicz83:Minimizing,Cox89:Heuristic}.
\cite{Fotakis22:Graph} studied the problem of establishing a spanning tree
using minimum queries. The authors studied noisy queries that may return wrong
edge states.

\subsection*{Advantage of model-free approaches}

{\sc RL} and {\sc MCTS} are model-free approaches so they generalise to more complex models.
Let us consider an example instance with interdependent edges, depicted in Figure~\ref{fig:interdependent}.
(Our paper assumes independent edges. Interdependent edges are only discussed in this section.)
For Figure~\ref{fig:interdependent}, {\sc H1} may start from any edge, because every edge is an intersection
of minimum cut and shortest path. Assume that {\sc H1} starts with $b$. In this case, there is no way
to guarantee that conclusion is reached within two queries. If the second query is $c$, then the third
query is always needed. If the second query is either $a$ or $d$, then there is $50\%$ chance of needing
the third query (i.e., when $a$ turns out to be {\sc Off}).
On the other hand, if we start with $a$ (or $d$), then we can guarantee to finish within two queries.
Due to symmetry, it is without loss of generality to assume that $a$ is {\sc On} and $d$ is {\sc Off} after the first query.
By querying $b$ next, we conclude regardless of $b$'s state. If $b$ is {\sc On}, then we have a path.
If $b$ is {\sc Off}, then we have a cut. In summary, for Figure~\ref{fig:interdependent}, there
is a correct way to start, which can easily be picked up by {\sc RL} and {\sc MCTS}, but heuristic like {\sc H1}
is unable to recognise this.

We leave model generalisation to future research.






\begin{figure}
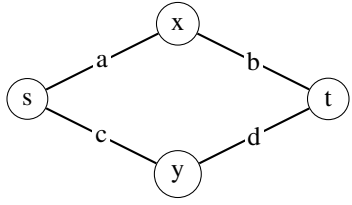

\centering
\captionsetup{justification=centering}
\tikz \graph [edge quotes={fill=white,inner sep=1pt},
          grow down, branch right, nodes={circle,draw}] {
    s --[thick, "a"] x[at={(2,2)}] --[thick, "b"] t[at={(4,2)}];
    s --[thick, "c"] y[at={(1,-1)}] --[thick, "d"] t;
         };
         \caption{Problem instance with {\bf interdependent} edges:\\
         $4$ undirected edges $\{a,b,c,d\}$\\
         Edge is {\sc On} with $p=0.5$\\
         Assume $a$ and $d$'s states are always opposite}
         \label{fig:interdependent}
\end{figure}

\end{document}